\documentclass[12pt]{article}
\usepackage{amsfonts}
\usepackage{alltt}
\usepackage{color}
\usepackage{graphicx}
\usepackage{amssymb,verbatim,ifthen}
\usepackage[all]{xy}
\usepackage{color}
\usepackage{graphicx}
\usepackage{amsmath}
\usepackage{amsthm}
\usepackage{mathabx}
\usepackage{pdfsync}
\usepackage{pgfplots}
\usepackage{natbib}
\usepackage{relsize}
\usepackage{xcolor}
\usepackage{hyperref}
\usepackage{mathtools}
\usepackage[inline]{enumitem}

\usepackage{comment}
\usepackage[textwidth=30mm]{todonotes}

\usepgfplotslibrary{polar}
\usepgflibrary{shapes.geometric}

\def\split{true}
\ifthenelse{\equal{\split}{true}} {
\textwidth = 440pt \oddsidemargin = 2pt \evensidemargin = 2pt
}
{}
\newcommand{\ignore}[1]{}

\renewenvironment{proof}[1][Proof.]{\textbf{#1} }{\ \rule{0.5em}{0.5em}}

\newtheorem{theorem}{Theorem}
\newtheorem*{theorem*}{Theorem}

\newtheorem{claim}{Claim}

\newtheorem{lemma}{Lemma}

\newtheorem{definition}{Definition}

\newtheorem{proposition}{Proposition}
\newtheorem{assumption}{Assumption}

\newtheorem{corollary}{Corollary}

\def\1{{1\hskip-2.5pt{\rm l}}}

\def\d{\delta}

\def\G{\Gamma}

\def\Pr{\textbf{{\rm Pr}}}

\def\N{\mathbb{N}}

\def\G1{G_{U\mbox-L}}

\def\SS{\mathcal{S}} 

\def\E{\mathbb{E}} 
\newcommand\ubar[1]{\underline{#1}} 

\pgfplotsset{my style/.append style={axis x line=middle, axis y line=middle, xlabel={$x$}, ylabel={$y$}, axis equal }}
\makeatother

\begin{document}
\title{I Want to Tell You?\\ [0.2em]\smaller{} Maximizing Revenue in First-Price Two-Stage Auctions\thanks{The authors wish to thank David Lagziel for his valuable comments.
G. Ashkenazi-Golan acknowledge the support of the Israel Science Foundation, grants \#217/17 and \#722/18, and NSFC-ISF, China Grant \#2510/17.
Y. Tsodikovich acknowledges the support of the French National Research Agency Grants ANR-17-EURE-0020.
The research of G. Ashkenazi-Golan and Y. Tsodikovich was partially supported by the COST Action CA16228 ``European Network for Game Theory'' (GAMENET).}}

\author{Galit Ashkenazi-Golan%
\thanks{The School of Mathematical Sciences, Tel Aviv
University, Tel Aviv 6997800, Israel. e-mail: \textsf{galit.ashkenazi@gmail.com} .},
Yevgeny Tsodikovich\thanks{Aix Marseille Univ, CNRS, AMSE, Marseille, France. e-mail: \textsf{yevgets@gmail.com}},
and Yannick Viossat%
\thanks{PSL, Université Paris–Dauphine, CEREMADE UMR7534, Paris, France}}

\maketitle

\begin{abstract}
\noindent A common practice in many auctions is to offer bidders an opportunity to improve their bids, known as a Best and Final Offer (BAFO) stage.
This final bid can depend on new information provided about either the asset or the competitors.
This paper examines the effects of new information regarding competitors, seeking to determine what information the auctioneer should provide assuming the set of allowable bids is discrete. 
The rational strategy profile that maximizes the revenue of the auctioneer is the one where each bidder makes the highest possible bid that is lower than his valuation of the item.
This strategy profile is an equilibrium for a large enough number of bidders, regardless of the information released. 
We compare the number of bidders needed for this profile to be an equilibrium under different information settings.
We find that it becomes an equilibrium with fewer bidders when less additional information is made available to the bidders regarding the competition.  
It follows that when the number of bidders is \emph{a priori} unknown, there are some advantages to the auctioneer to not reveal information.
\end{abstract}
\bigskip
\noindent {\emph{Journal} classification
number: D44, D82.}

\bigskip


\newpage

\section{Introduction}
This paper examines information effects in two-stage auctions involving a large number of bidders.
It is intuitively clear and theoretically understood that, under rather general assumptions, when large numbers of bidders participate in an auction, the strong competition drives the revenue of the auctioneer up.  
\cite{wilson1977bidding}, followed by \cite{milgrom1979convergence}, showed that in a first-price sealed-bid auction with affiliated values, as the number of bidders increases, the winning bid converges to the value the winner attributes to winning the auction.
This mechanism is not only efficient, but leaves the entire economic surplus in the hands of the auctioneer.\footnote{More about efficiency in large auctions can be found in \cite{pesendorfer2000efficiency}, \cite{swinkels2001efficiency}, \cite{bali2002asymptotic}, \cite{fibich2010large} and the references therein.
}

The importance of the number of bidders to the revenue of the auctioneer was highlighted by \cite{bulow1994auctions}, who showed that the additional revenue provided by one extra bidder could exceed the revenue from choosing the optimal allocation mechanism. 
Another reason for analyzing auctions with many bidders was stated in \cite{wilson1985incentive} and \cite{neeman2003effectiveness}:
it allows for lighter assumptions regarding bidders' knowledge of features such as the exact number of bidders or the distribution of their valuations.
Often, bidders only need to know that there is a large number of participants in the auction.
Furthermore, in many settings the exact equilibrium strategies are hard to compute, whereas the asymptotic strategies are more straightforward, as illustrated by \cite{bali2002asymptotic} among others.

In this paper we analyze two-stage auctions.
It is common to have two stages in procurement auctions, especially those relating to construction contracts. 
In the first stage, the bidders submit sealed bids that are evaluated and the highest ones are chosen.
These highest bidders move on to the second stage, known as the Best And Final Offer (BAFO) stage, where they are allowed to increase their bids.
The winner in the auction is chosen based on the BAFO-stage bids, and typically pays his bid (first-price auction).

This paper was inspired by a consulting job.
An auctioneer\footnote{Hereafter, as customary in auctions, the auctioneer will be referred to as a female ``seller'' and the bidders as male ``buyers'', although in procurement auctions the roles are reversed.} was seeking advice on what information to provide to the bidders before the BAFO stage in order to maximize her revenue.
In general, such information can help bidders learn either their private values or the level of competition.
While an extensive literature has examined the former information effect (e.g.,  \cite{pesendorfer2000efficiency, swinkels2001efficiency}), little is known about the latter.
Our paper fills this gap in the literature by considering which information, if any, should be given to bidders between stages to maximize revenue.
To isolate the problem of how to stimulate competition via information, we chose a model of private valuations.
Naturally, if there are common unknown components to the valuation, additional information teaches each bidder not only about the competition but also about his private value, which can significantly affect future bids.
Hence, in our model, information about the other bids and bidders does not affect valuations.
It only affects beliefs regarding the second-stage bids and how competitive the rest of the bidders might be.

While the literature suggests that a second-price sealed-bid auction is optimal (in terms of revenue for the auctioneer) for our setting \citep{myerson1981optimal}, this type of auction is not widely used.
Two-stage auctions, however, are common.
The optimal two-stage auction design when the values might be correlated was explored in \cite{perry2000sealed}, that considered a model where the second stage is a second-price sealed-bid auction.
The optimal information to be released by the auctioneer in a two-stage auction where the second stage is a second-price sealed-bid auction was investigated by \cite{ganuza2007competition}.
Our research is restricted to the more common auction structure, where the second stage is a first-price sealed-bid auction.
As far as we know, this model, particularly the question of the optimal information structure, has not previously been explored.

In our model, $n$ bidders with 
private values participate in a sealed-bid two-stage auction with discrete\footnote{We discuss the continuous case in Appendix \ref{app:cont_auc}.} bids.
The top two bidders from the first stage proceed to the second stage (the BAFO stage), where they can either increase their bid or maintain the previous one.
The highest bidder in the second stage wins the auction, pays this bid, and receives the good.
Between the stages, the seller can privately send a message to each of the finalists regarding others' bids, their ranking, or other relevant information.

This information can trigger two opposing effects.
On the one hand, information about competitors can drive  prices up, as bidders seek to increase their chances of victory.
On the other hand, the possibility of future information can reduce bids in the first round.
A bidder may choose to submit a lower bid in the first stage and increase it only if he finds himself facing a strong competitor.
Otherwise, the smaller initial bid becomes the final one, decreasing the revenue of the seller.

The main goal and contribution of this paper are to provide a theoretical answer to a practical question, and to deduce the conditions under which each of the above-mentioned effects is stronger.
More precisely, we seek to determine when the strategy profile that maximizes the revenue of the seller is an equilibrium, and when it is the unique equilibrium for different information structures.

We show that the revenue-maximizing strategy profile is an equilibrium for a large enough number of bidders, regardless of the information released between the stages.
We compare the number of bidders required for this strategy profile to be an equilibrium across different information structures.
We find that the less information\footnote{Information structures are pair-wise compared using an adaptation of \cite{blackwell1951comparison} to our model. See Definition \ref{def:more_inf}.} is given to the bidders, the fewer bidders are required to maintain this strategy profile as an equilibrium.
As a result, when the number of bidders is unknown, there are some advantages to conducting an auction without revealing information before the BAFO stage.

This result may appear to be in conflict with the linkage principle (\cite{milgrom1982theory} and \cite{milgrom1985economics}) which states that committing to reveal information before or along the auction benefits the auctioneer (even if the information is ``bad'').
As argued by \cite{milgrom1982theory}, revealing information decreases the uncertainty of the participants in the auction regarding the actual value of winning the auction, thus mitigating the winner's curse and leading to more competitive bidding.
In our model the bidders are certain regarding their valuation, and therefore there is no contradiction between the linkage principle and the advantage we find in releasing no information.

We conjecture not only that it takes fewer bidders to obtain the revenue-maximizing equilibrium in auctions where information is withheld than in auctions with information, but also that it remains the unique equilibrium under fewer bidders when less information is provided.
This strengthens the claim that providing additional information between stages is not to the advantage of the seller, as it gives bidders an opportunity to play an equilibrium where they bid less.
As is frequently the case (see, for example, \cite{quint2018theory}), the uniqueness of this equilibrium can only be proven in special cases.
We provide such a proof in Appendix \ref{app:uniquness} for a simple scenario and test this hypothesis for more complicated scenarios using a computer simulation.

Our analysis differs from common practice in the auction literature.
Instead of finding equilibrium strategies in a given auction or finding the optimal mechanism for specific settings, we solve an implementation problem. 
In our work, a particular strategy profile, the revenue-maximizing profile, is desired and we construct an information structure in which this profile is an equilibrium in a two-stage auction with discrete bids and many bidders.
For a small number of bidders, multiple equilibria can exist, including equilibria with mixed actions or relatively low bids.
The example in Section \ref{sect:smallnumberofbidders} shows that the optimal information structure in an auction with a small number of bidders depends on the fine details of the value distribution.

The implication of our results, at least for a large number of bidders, is that information provided before the BAFO stage might decrease the revenue of the seller.
Hence, when the sole purpose of the BAFO stage is to encourage competition, providing no  information tends to be better for the seller.
It also renders the BAFO meaningless.
Since the bidders learn nothing between stages, the optimal strategy is to bid the same amount in both stages.
In such situations, there is an argument for avoiding the BAFO stage altogether, which can reduce costs and prevent delays \citep{ahadzi2001private,dudkin2006transaction}.

The remainder of the paper is organized as follows.
This introduction is followed by a short survey of related literature and findings regarding multi-stage auctions.
Section \ref{Sect:model} formally presents the model and the various information structures.
In Section \ref{two_lemmas} we present two useful equivalence lemmas, which help simplify the analysis of auctions where information is not revealed.
Section \ref{example} demonstrates the model through two examples and sets the ground for the main results, which are presented in Section \ref{results}.
Section \ref{Sect:Discussion} includes the concluding remarks and a discussion of the remaining open questions.
A partial answer regarding the uniqueness of the equilibrium is given in Appendix \ref{app:uniquness}.
Finally, in Appendix \ref{app:otherassumptionFn} we present a more general but less intuitive assumption regarding the joint distribution of the private values that remains consistent with our results and in Appendix \ref{app:cont_auc} we discuss auctions with a continuum of private values and possible bids.
We show that in the continuous case, as in the discrete, there are good reasons for not revealing information.

\subsection{Motivation and Related Literature}
While commonly used, two-stage auctions with a BAFO stage have received little research attention. 
For example, the Federal Transit Administration recommends a BAFO stage in its manual\footnote{Best Practices~Procurement Manual, Federal Transit Administration, November 2001.} as a means to conclude the auction and receive improved offers.
A BAFO stage is also commonly used in the UK and France as the final stage of procurement auctions \citep{noumba2005private}.
It appears in large-scale infrastructure projects around the world, such as the water distillation factory in Ashkelon, Israel  \citep{SAUVETGOICHON200775} and in prison procurement and operations in South Africa \citep{merrifield2002asset}.
It was even used as a tie-breaking stage when the auctioneer could not decide between two offers in terms of highest value for money, and decided to implement an additional (unplanned) BAFO stage \citep{rintala2008organizing}. 
Furthermore, it is the US government's preferred method of buying goods, as suggested in \cite{chelekis1992official} and illustrated in \cite{roll2000online,finley2001dynamic} and others.

The World Bank, in a procurement guide entitled ``Negotiation and Best and Final Offer (BAFO)''\footnote{\href{http://pubdocs.worldbank.org/en/663621519334519385/Procurement-Guidance-Negotiation-and-Best-Final-Offer.pdf}{http://pubdocs.worldbank.org/en/663621519334519385/Procurement-Guidance-Negotiation-and-Best-Final-Offer.pdf}} states that:
\begin{quote}
``BAFO is appropriate when the procurement process may benefit from Bidders/Proposers having a final opportunity to improve their Bid/Proposal, including by reducing prices, clarifying or modifying their Bid/Proposal, or providing additional information.
It is normally particularly effective when markets are known to be highly competitive and there is strong competitive tension between Bidders/Proposers.''
\end{quote}
Then the guide lists the objectives of the BAFO -- the two main aims being to increase understanding by bidders of the auctioneer's requirements and to enhance competition among bidders who have made a proposal.
The results presented here mainly relate to the second aim, addressing the following question: does the BAFO enhance competition?
We also investigate its effectiveness when markets are highly competitive, as suggested by the World Bank's guide.

Another possible reason to employ two-stage auctions is for screening purposes: the auctioneer screens out bidders that do not meet some minimum quality or price standards, and in the second stage, she runs an English auction among the remaining bidders.
Thus, the purpose of the first stage is to ensure that the good provided meets the minimum requirements, and the only remaining issue, the price, is settled in the second stage.

Screening can also serve the purpose of limiting competition to encourage bidders to participate in an auction with entry costs.
For example, in indicative bidding \citep{ye2007indicative, quint2018theory} the seller asks bidders to submit non-binding first-stage bids to evaluate their interest in the good.
She then chooses the best bidders to proceed to the second stage, in which they prepare their ``real'' bid.
Since preparing the second-stage bid is costly (requires time, effort, learning about the asset, and so on), it is important to reassure the bidders that their chances of winning are high enough to justify the entry cost.
This cost also discourages bidders from bidding high just to qualify, as they are not sure that they want to participate in the second stage.

Multiple stages can also help bidders learn their rivals private value.
This happens when the bidder does not know for certain his own valuation (but rather has an estimation of this value given his signal) and the values are correlated (as in \cite{milgrom1982theory}).
In this case, the information about bids of others changes the bidder's own evaluation of the gain in winning the auction.
In such settings, a sealed-bid first-price auction is significantly different from an English ascending auction.
In the former, bids are placed according to initial evaluations of the values, there is no updating, and the chances of overbidding and suffering from the ``winner's curse'' are higher.
In contrast, the constant updating of a bidder's estimation of the value during an English auction, based on the current price and the number of buyers that quit at each level, can reduce bids and prevent such problems \citep{perry2000sealed}.

Similarly, two-stage auctions can be used to extract information in asymmetric settings, when there is one bidder more informed than the others (and possibly the seller), such as in the case of privatization auctions \citep{caffarelli1998auction,perry2000sealed,dutra2002hybrid,ncc2007}. 
This can work even with non-binding first stage bids \citep{boone2009optimal}, since the more informed bidder has an incentive not to bid too low and be eliminated when the asset is costly and by doing so -- reveal some of his information.
In particular, when the second stage is a second-price auction between the two highest bidders, \cite{hernando2018inefficient} showed that in equilibrium the more informed bidder submits his private value in both stages.

Since we are studying the effect of information on bidding competition, we assume that private values are known from the outset of the auction.
Bids are therefore updated in the second round for the sole purpose of increasing the chances of winning.
This can occur when the bidder learns the other bids and realizes that he is about to lose.
For example, the following question appears in the FAQ section of the website of the US Federal Transit Administration\footnote{\href{https://www.transit.dot.gov/funding/procurement/third-party-procurement/best-and-final-offer}{https://www.transit.dot.gov/funding/procurement/third-party-procurement/best-and-final-offer}}:
\begin{quote}
Q: Is it permissible for an offeror to lower its price for the Best and Final Offer (BAFO) without any basis for the change, other than trying to beat out the competition, or does the price reduction have to be based on changes or other clarifications discussed during the presentation?
We are participating as an offeror and have been asked for a BAFO.

A: Yes, the basis for change may be solely the desire to increase your chances of winning the contract award by lowering your price. If the contract you are competing for is a cost-reimbursement type contract, the procuring agency may well ask you for your rationale in lowering the original cost estimates to do the work. [...]
\end{quote}
In theory, a bidder could also increase his bid without any additional information.
In Lemma \ref{thmNoInfoIsFPSA} we show that this kind of bidding is sub-optimal.
If it is profitable for the bidder to submit a better offer in the second stage without additional information, he might just as well submit it in the first stage and increase his chances of winning, as nothing changes between the stages.


\section{The Model}\label{Sect:model}
We consider a two-stage auction with private valuations.
In the first stage, all bidders submit their bids. 
The two highest bidders proceed to the second stage, called Best-And-Final-Offer (BAFO) stage.
The first-stage bids are binding in the sense that each participant in the BAFO stage must bid at least as high as his first-stage bid. 
The highest bidder at the BAFO stage wins the auction and pays his second-stage bid.
In both stages, we use a symmetric tie-breaking rule.

Formally, let $N=\left\{1,\ldots,n\right\}$ be the set of $n\geq 3$ bidders, $V=\left\{v^1,\ldots,v^K\right\}$ the set of possible private valuations (types) in ascending order and $F_n$ the joint distribution assigning types to bidders.
In Assumption \ref{uniformfullsupport} we explain how $F_n$ changes with $n$.
This assumption is slightly generalized in Appendix \ref{app:otherassumptionFn}.

After observing his own valuation, bidder $i$ places his first-stage bid, $b_i^1$, chosen from a given set of allowable bids $B=\left\{ b^1,\ldots,b^M\right\}$ (in ascending order).
Presetting a set of allowable bids is a common practice, as the auctioneer prefers bids to be ``rounded'' or in fixed increments; moreover,  this is a natural limitation as the currency is not continuous.
We assume that the set $B$ is rich enough so that for any $v^j\in V$, there is a bid strictly between $v^{j-1}$ and $v^j$ (for this matter, $v^0=-\infty$).
The quantities $\beta^j=\max\left\{b\in B\vert b<v^j\right\}$ are thus well defined.
\footnote{If we remove the richness of allowable bids assumption, the only thing lost is the efficiency of the auction -- the item might no longer go to the bidder who values it highest, since the bids sometimes cannot differentiate between bidders with different private values.}

The two\footnote{Our results can be generalized to cases where more than two bidders proceed to the BAFO stage.} bidders who bid highest in the first stage move on to the BAFO stage.	
These two bidders receive information from the auctioneer regarding the results of the first stage before the second stage starts.
The information is a function from the set of first-stage bids to a set of vectors of messages (one for each bidder).
Let $\Theta_i$ be the set of possible messages bidder $i$ may receive.
The information structure is $\Theta: B^n\rightarrow \Delta\left(\Theta_1\times\ldots\times \Theta_n\right)$, where as usual, for any finite set $X$, $\Delta(X)$ denotes the set of probability distributions over $X$.
The realized vector of messages is denoted by $(\theta_1,\ldots,\theta_n)\in \Theta_1\times\ldots\times \Theta_n$.

In the second stage, the two bidders place a second bid (once again, within $B$), the only restriction being that their second bid must be at least as high as their first bid.
Hence, a behavioral strategy of bidder $i$ is a function $\sigma_i=(\sigma_i^1,\sigma_i^2)$ such that $\sigma_i^1: V\rightarrow \Delta(B)$ and $\sigma_i^2: V\times B\times \Theta_i \rightarrow \Delta(B)$.
The strategy profile of all bidders is denoted by $\sigma=\bigtimes\limits_{i\in N} \sigma_i$ and for each $i$, the strategy profile of all bidders except $i$ is denoted by $\sigma_{-i}=\bigtimes\limits_{\substack{j\in N \\ j\neq i}} \sigma_j$.
The highest bidder at the second stage wins the auction (ties are broken randomly with equal probabilities).
The utility for the winner is the difference between his valuation and the amount he pays -- his second-stage bid.

We wish to compare information structures. 
To this end, we adjust the partial order over experiments introduced by \cite{blackwell1951comparison}, which reflects the amount of information provided by the experiment.
Intuitively, an experiment $A$ is more informative than an experiment $B$ if there exists a garbling function from the results of experiment $A$ to the results of experiment $B$ which generates the same distribution of results as experiment $B$.
For example, flipping a coin and observing the result is more informative than letting someone else do it and then report the truth with some probability and lie otherwise.
This second step of telling the truth or lying is exactly the above mentioned garbling function.

\begin{definition}\label{def:more_inf}
The information structure $\Theta: B^n\rightarrow \Delta\left(\Theta_1\times\ldots\times \Theta_n\right)$, is \emph{more informative than} the information structure $\Theta': B^n\rightarrow \Delta\left(\Theta_1\times\ldots\times \Theta_n\right)$, 
if for every bidder $i\in N$, and every strategy profile of the others $\sigma_{-i}$, there exists $L_i:\Theta_i\times B \rightarrow \Delta(\Theta_i)$, such that for all $\theta_i'\in \Theta_i$ and all first-stage bids of bidder $i$, $b_i\in B$,
\begin{equation}\label{eq_more_info}
\sum_{b_{-i}\in B^{n-1}}\Pr_{\sigma_{-i}}(b_{-i})\Pr_{\Theta'}(\theta_i'\vert b_{-i},b_i)=
\sum_{b_{-i}\in B^{n-1}}\sum_{\theta_i\in \Theta_i}\Pr_{\sigma_{-i}}(b_{-i})\Pr_{\Theta}(\theta_i\vert b_{-i},b_i)\cdot \Pr(L_i(\theta_i,b_i)=\theta_i').
\end{equation} 
\end{definition}

A more natural definition might be to assume that $L_i$ depends solely on $\theta_i,b_i$ and is independent of the strategy of the others $\sigma_{-i}$.
This is similar to other adaptations of \cite{blackwell1951comparison} to repeated games with imperfect monitoring, such as in \cite{lehrer1992equilibrium,hillas2016correlated}, and \cite{ashkenazi2019blackwell}.
All our results hold under this stronger definition, since it implies Definition \ref{def:more_inf}.
However, the stronger definition is less applicable as it allows less information structures to be pair-wise compared relative to Definition \ref{def:more_inf}.
Moreover, when discussing the equilibrium strategies or a non-informative information structure (Definition \ref{def:non_info} below), we assume that $\sigma_{-i}$ is known to bidder $i$, so fitting a garbling function to each strategy $\sigma_{-i}$ does not require additional information that bidder $i$ does not have.

In this paper we compare more informative information structure to less informative ones.
In particular, we compare \emph{informative} information structures to \emph{non-informative} ones.\footnote{Often referred to as ``with information'' and ``without information'', respectively.}
In the context of this paper, the bidders may use their information when deciding on their second-stage bids.
In equilibrium, a bidder optimizes his gain at the second stage given the joint distribution of values, his value, his first-stage bid, his opponents' strategies, and the signal he observes between the stages. 
Given these, a bidder who proceeds to the second stage rationally forms his beliefs regarding his opponent's bid at the second stage based on Bayes' Rule, and then he best responds to these beliefs.
An information structure is non-informative if all signals conveyed with positive probability do not affect these beliefs, whatever the strategies of the players are. 

\begin{definition}\label{def:non_info}
Suppose that bidder $i$ with private value $v_i$ is selected to proceed to the second stage after biding $b_i^1$, and let $B_i^{opp}$ be the random variable that represents the bid of the other bidder in the second stage. 
An information structure $\Theta$ is \emph{non-informative} if the conditional distribution of $B_i^{opp}$ does not change due to the information relayed to bidder $i$. 
More precisely, if
\begin{equation}
\forall F_n,\forall \sigma_{-i},\forall \theta_i\in \Theta_i, \forall b\in B: \Pr(B_i^{opp}=b|F_n,v_i,b_i^1,\sigma_{-i},\theta_i)=\Pr(B_i^{opp}=b|F_n,v_i,b_i^1,\sigma_{-i}),
\end{equation}
where 
$v_i\in V$ is the private value of bidder $i$, and  $\sigma_{-i}$ is the strategy profile of everyone but bidder $i$.
We implicitly assume that we condition only on events of non-zero probability, so given $F_n,v_i,b_i^1,\sigma_{-i}$, bidder $i$ has a non-zero probability to proceed to the second stage (and it happened) and the $\forall \theta_i\in \Theta_i$ part refers only to messages with non-zero probability given the first-stage bids.

Otherwise, the information structure is \emph{informative}.
\end{definition}

A non-informative information structure is such that the message from the auctioneer does not change the belief of bidder $i$ regarding the distribution of the second stage bids of the other finalist.
The simplest non-informative information structure is that where for all $i$, $\Theta_i$ is a singleton and the same message is conveyed to the bidders regardless of the actual bids.
In Lemma \ref{lem:non_info_iff_iid} we present a simplifying characterization of non-informative information structures by showing that an information structure is non-informative if and only if the information to bidder $i$ is independent of the first-stage bids of the others (it can depend on his own bid).

In an informative information structure, there is a \emph{possibility} that the information changes this belief.
From Definition \ref{def:more_inf}, it is clear that the most informative information structure is the one that reveals all the first-stage bids to the finalists, which is all the information that the seller has.
We refer to this structure as \emph{the fully-informative} information structure.
Given the bid of bidder $j$, bidder $i$ can update his belief about $j$'s private value and, combined with $\sigma_{-i}$, his belief regarding $j$'s second-stage bid.

What we are examining here is the effect of the number of bidders and the information structure on equilibrium strategies and equilibrium payoff.
Therefore, and although Lemma \ref{thmNoInfoIsFPSA} and Theorem \ref{thmEFFEqWithIsAlsoWithout}  are true in general, 
the solution concept we study is symmetric equilibrium, i.e. an equilibrium where all bidders with the same private value use the same strategy.
We also should specify how the joint distribution of valuations changes when the number of bidders change, that is, how $F_n$ changes with $n$.
Clearly, this is not an issue if the valuations are i.i.d.
Otherwise, it is essential to add an assumption on the conditional distribution of the other valuations given a particular bidder who has the highest private value.

Let $F_n$ be some distribution and for every  $k\in\{1,\ldots,K\}$, we denote by $D_k$ the event that the types of all bidders are at most $v^k$.
Given  $D_k$, and given that the valuation of a specific bidder $i$ is $v^k$,  we denote by $\d^k(F_n)$ the minimal probability of each other bidder having the type $v^k$.
The minimization is done in the following manner: 
given $D_k$ and given the valuations of every possible subset of bidders that does not include bidder $j$, the latter has a probability of at least $\d^k(F_n)$ of having the type $v^k$.
The uniform bound over all $k$ is defined by $\d(F_n)\coloneqq\min_k \d^k(F_n)$.

For our results to hold valuations can be correlated, but not \emph{too correlated} in the sense that type $v^k$ cannot be ruled out for bidder $j$, regardless of the types of all the other bidders.
For example, when the valuations are i.i.d. (with full support over $V$), $\d^k(F_n)=\tfrac{p^k}{\sum\limits_{j\leq k} p^j}$, where $p^j$ is the probability of each bidder having type $v^j$, and $\d(F_n)=\min\limits_k\tfrac{p^k}{\sum\limits_{j\leq k} p^j}$.

We assume that as $n$ increases, these lower bounds become strictly positive and independent of $n$:
\begin{assumption}\label{uniformfullsupport}
The family of distributions $F_n$ has \emph{uniform full support}, i.e.
\begin{equation}
\d\coloneqq\liminf\limits_{n\to\infty} \d(F_n)>0.
\end{equation}
\end{assumption}
As demonstrated above, this condition is fulfilled when the valuations are i.i.d. with full support over $V$, since $\d=\d(F_n)=\min\limits_k\tfrac{p^k}{\sum\limits_{j\leq k} p^j}>0$.

We condition on $D_k$ to focus on the bidders with the highest realized type since typically, when the number of bidder is large or when the bids are increasing with the valuation, in equilibrium, a bidder with high private value always outbids a bidder with low private value, so a bidder with a private value $v^k$ (for $k<K$) has a positive probability to win only conditioning on $D_k$.
We formalize this reasoning in the Results section.

The intuition behind the uniform full support assumption can be taken from the following pathological example.
Assume $F_n$ is such that it assigns to one of the bidders the highest type and to all the others the lowest type.
The high-type bidder knows that he is the only one, and can bid a very low bid that is still higher than the value of all the other bidders.
This remains true regardless of $n$ or the design of the auction.
To ensure that this high-type bidder feels a competitive pressure to increase his bid as the number of bidders increases, the expected number of other high-type bidders should not go to $0$ with $n$.
This is the case under the uniform full support assumption, since the expected number of type $k$ bidders is at least $n\d$.

Note that the condition on $F_n$ can be slightly weaker (but also less elegant).
A variant of this assumption can be found in Appendix \ref{app:otherassumptionFn}.
This concludes the description of the model at the bidders' part.

As for the seller, she utilizes her role as a mechanism designer that chooses an information structure to maximize her expected payoff in the resulting auctions.
We assume that the bidders adopt equilibrium strategies in the auction based on the information structure and neither use weakly dominated strategies nor bid their valuation or above.\footnote{For bidder $i$, bidding $v^i$ or higher in the first round is weakly dominated. 
Bidding $v^i$ or higher in the second round is also weakly dominated, unless player $i$ knows that some other player $j$ made a first-round  bid $b_j^1 \geq v_i$.
In that case, there no risk for bidder $i$ in bidding $b_1^2= v_i$ (if this is an allowable bid) or even $b_i^2 > v_i$ as long as $b_i^2 < b_j^1$.
Allowing for such behaviors would however lead to a rather spurious multiplicity of equilibria.}
Hence, she wishes each bidder to bid his highest possible bid, i.e. his corresponding $\beta^i$ in both rounds.
Such a strategy will be referred to as the \emph{revenue-maximizing strategy} and denoted by $\sigma^*$.
The strategy profile in which all bidders use the revenue maximizing strategies is called the \emph{revenue-maximizing profile} and denoted by $\Sigma^*$.
The main questions that we explore in this paper is when $\Sigma^*$ is indeed an equilibrium and when it is the unique one (minding the restriction to non-weakly dominated strategies that never bid above the private value). 


\section{Equivalences in Auctions Without Information}\label{two_lemmas}
We start our analysis with two useful lemmas regarding non-informative information structures, which significantly simplify their analysis.
First, we show that the signals conveyed to bidder $i$ in a non-informative information structure cannot depend on the bids of the others.
In Definition \ref{def:non_info} we only prohibited them from depending on the distribution of the second stage bid of the opponent.
In theory this still allowed them to depend on low (``losing'') bids.
We show that is not the case.

A direct result of this independence is Lemma \ref{thmNoInfoIsFPSA}, which states that in a non-informative information structure, a strategy that changes the bids on the second stage is weakly dominated by one that does not.
Hence, when considering only equilibria in non-weakly dominated strategies, it is possible to assume that bidders do not change their bids in the second stage, which  makes this stage redundant.
This significantly simplifies the analysis, as it effectively turns this auction into a single-stage first-price auction.

\begin{lemma}\label{lem:non_info_iff_iid}
An information structure $\Theta$ is non-informative iff for every bidder $i$ with private value $v_i$ that bids $b_i^1$ in the first stage, the distribution of the $i$th element of $\Theta$, given $b_i^1$ and conditioned on bidder $i$ reaching the second stage, is independent of all the other bids.
\end{lemma}
\begin{proof}
The first direction of the lemma is trivial.
If the information is independent of the first-stage bids, conditioning on it will not affect the second-stage bids of the opponent.
The rest of the proof deals with the second direction.

\smallskip

Let $\Theta $ be a non-informative information structure and assume by contradiction that there exists a bidder $i$ with private value $v_i$ that the signals he receives when reaching the second stage depend on the first-round bids of the others.
This means that there exists a first stage bid $b_i^1$ and at least two vectors of bids of the others such that in both cases:
\begin{enumerate*}
\item the probability bidder $i$ reaches the second stage (the event $\SS$) is positive;
\item the distribution of signals received by bidder $i$ is different.
\end{enumerate*}
Without loss of generality, assume the lowest bid is $b^1=0$ and denote by $\ubar{0}_{-i}$ the vector of bids where everyone bid $0$.
Clearly $\Pr(\SS\vert b_i^1,\ubar{0}_{-i})>0$.
Let $\ubar{b}_{-i}$ another vector of bids such that $\Pr(\SS\vert b_i^1,\ubar{b}_{-i})>0$ and the distribution of signals bidder $i$ receives when everyone bids according to $\ubar{b}_{-i}$ is different than when everyone bids according to $\ubar{0}_{-i}$.
Thus, there exists a signal $\theta_i\in\Theta_i$ such that $\Pr(\theta_i\vert \ubar{0}_{-i},b_i^1)>\Pr(\theta_i\vert \ubar{b}_{-i},b_i^1)$.

Consider the following joint-distribution $F_n$ (conditional on bidder $i$ having value $v_i$):
\begin{itemize}
\item With probability $\tfrac{1}{2}$ all other bidders have the private value $v^1$.
\item With probability $\tfrac{1}{2}$ all other bidders have the private value $v^K$.
\end{itemize}
As for the strategy, each of the bidders $j\neq i$ bids in the second stage the same bid he bid in the first stage.
In the first stage, bidder $j$ with private value $v_j$ bids $0$ if $v_j=v^1$ and bids his corresponding bid in $\ubar{b}_{-i}$ if $v_j=v^K$.

Denote the maximal bid in $\ubar{b}_{-i}$ by $\tilde{b}$.
To achieve the contradiction, we calculate the probabilities of the events $B_i^{opp}=0$ and $B_i^{opp}=\tilde{b}$ conditioned and unconditioned on $\theta_i$ and show that conditioning on the signal $\theta_i$ changes the distribution of $B_i^{opp}$.

When bidder $i$ bids $b_i^1$, the probabilities of the two possible second stage bids conditioned on the signal $\theta_i$ (the conditioning on $F_n$, $b_i^1$, and the strategy of the others is implicit) are
\[\Pr(B_i^{opp}=0\vert \theta_i, \SS)=
\tfrac{\Pr(\ubar{0}_{-i}, \theta_i, \SS)}{\Pr(\theta_i,\SS)}=
\tfrac{\Pr(\ubar{0}_{-i})\Pr(\SS\vert\ubar{0}_{-i})\Pr(\theta_i\vert \ubar{0}_{-i})}{\Pr(\theta_i,\SS)},\]
and
\[\Pr(B_i^{opp}=\tilde{b}\vert \theta_i, \SS)=
\tfrac{\Pr(\ubar{b}_{-i}, \theta_i, \SS)}{\Pr(\theta_i,\SS)}=
\tfrac{\Pr(\ubar{b}_{-i})\Pr(\SS\vert\ubar{b}_{-i})\Pr(\theta_i\vert \ubar{b}_{-i})}{\Pr(\theta_i,\SS)}.\]

Without receiving the signal $\theta_i$, the probability of seeing each of the possible bids in the second stage are
\[\Pr(B_i^{opp}=0\vert \SS)=
\tfrac{\Pr(\ubar{0}_{-i}, \SS)}{\Pr(\SS)}=
\tfrac{\Pr(\ubar{0}_{-i})\Pr(\SS\vert\ubar{0}_{-i})}{\Pr(\SS)},\]
and
\[\Pr(B_i^{opp}=\tilde{b}\vert \SS)=
\tfrac{\Pr(\ubar{b}_{-i}, \SS)}{\Pr(\SS)}=
\tfrac{\Pr(\ubar{b}_{-i})\Pr(\SS\vert\ubar{b}_{-i})}{\Pr(\SS)}.\]
By writing these terms explicitly, it is easy to see that
\begin{equation}\label{eq:ratio_depends_on_theta}
\tfrac{\Pr(B_i^{opp} = 0 \vert \theta_i, \SS)}{\Pr(B_i^{opp} = 0 \vert \SS)}>\tfrac{\Pr(B_i^{opp} = \tilde{b} \vert \theta_i, \SS)}{\Pr(B_i^{opp}=\tilde{b} \vert \SS)}.
\end{equation}
This strict inequality implies that the distribution of $B_i^{opp}$ is not independent of $\theta_i$ (otherwise both ratios would be equal to $1$), which contradicts the fact that $\Theta$ is non-informative.\hfill\end{proof}

Using this lemma we can establish that the least informative information structures are the non-informative ones.
Indeed, suppose that $\Theta'$ is a non-informative information structure.
According to the lemma, the signals that bidder $i$ receives are independent of the bids of the others, so for each $\theta_i'\in \Theta_i$, the left-hand side of Eq.~\eqref{eq_more_info} is $\sum_{b_{-i}\in B^{n-1}}\Pr_{\sigma_{-i}}(b_{-i})\Pr_{\Theta'}(\theta_i'\vert b_{-i},b_i)=\Pr_{\Theta'}(\theta_i'\vert b_i)$.
For each information structure $\Theta$ we can define $L_i$ according to $L_i(\theta,b_i)=\Pr_{\Theta'}(\theta_i'\vert b_i)$ (independent of $\theta$) for which the right-hand side of Eq.\eqref{eq_more_info} reduces to $\Pr_{\Theta'}(\theta_i'\vert b_i)$, and the equation holds.

A direct result of this lemma is that in a non-informative information structure, the bidders do not need to wait for the information given by the seller if their second stage depends on it.
After choosing his first-stage bid, a bidder already knows the distribution of the signals he is going to obtain (as it is independent of the actions of the other bidders), so he can mimic this information structure himself and select his second stage bid before starting the auction.
By submitting this (higher) second stage bid already in the first stage (and not changing it afterwards), he increases his chances of being selected to the second stage without affecting the expected payoff if selected.
This logic, which makes the second stage redundant, is summarized in the following lemma.

\begin{lemma}\label{thmNoInfoIsFPSA}
If $\Theta$ is non-informative, then each strategy in which the second bid differs from the first is weakly dominated by a strategy in which the second bid equals the first.
\end{lemma}
\begin{proof}
Fix $F_n$.
Consider bidder $i$ with a private value $v_i$.
Let $\sigma_i=(\sigma_i^1,\sigma_i^2)$ be some strategy of bidder $i$.
We prove that there exists a strategy $\hat{\sigma}_i=(\hat{\sigma}_i^1,\hat{\sigma}_i^2)$ s.t. $\forall b\in B, \hat{\sigma}_i^2(v_i,b,\cdot)=b$ which weakly dominates $\sigma_i$.

Without loss of generality, we can assume that $\sigma_i$ never bids above $v_i$.
Otherwise, we can consider a new strategy which bids $\beta^i$ whenever $\sigma_i$ instructs to bid higher than $v_i$, and otherwise bids the same as $\sigma_i$.
This new strategy replaces non-positive payoffs with non-negative ones, so its expected utility is as good as of $\sigma_i$.
By applying the rest of the proof to this new strategy, we increase the expected payoff even further.

If the rest of the bidders use the strategy $\sigma_{-i}$, the probability of bidding $b_i^2$ in the second stage is\footnote{As for convention, when conditioning on a zero-probability event, the whole product is zero.}
\begin{equation}\label{eq:totalprob_b_i^2}
\Pr_{\sigma_i}(b_i^2\vert v_i)=\sum\limits_{\substack{b_i^1\leq b_i^2 \\ \theta_i\in \Theta_i \\ b_{-i}\in B^{n-1}}} \Pr(\sigma_i^2(v_i,b_i^1,\theta_i)=b_i^2\vert \theta_i, b_i^1) \Pr(\theta_i\vert b_i^1, b_{-i})\Pr(b_{-i}\vert F_n, \sigma_{-i})\Pr(\sigma_i^1(v_i)=b_i^1).
\end{equation}
Conditioned on $b_i^1$ the information is independent of the bids of the others so $\Pr(\theta_i\vert b_i^1, b_{-i})=\Pr(\theta_i\vert b_i^1)$.
In addition, $\sum\limits_{b_{-i}\in B^{n-1}}\Pr(b_{-i}\vert F_n, \sigma_{-i})=1$ so the dependence on $\sigma_{-i}$ and $F_n$ vanishes from the above expression, and the probability of bidding $b_i^2$ in the second stage is uniquely determined by $v_i$.
We can therefore define a strategy $\hat{\sigma}_i^1(v_i)$ to be a a distribution over $B$ that chooses a bid according to the above probabilities and $\hat{\sigma}_i^2$ to be a strategy that chooses the first-stage bid with probability $1$.
It is left to show that  $\sigma_i$ is weakly dominated by  $\hat{\sigma}_i$.

The expected payoff when using the strategy profile $\sigma_i$ is
\begin{eqnarray}
&\sum\limits_{b_i^2\in B}\sum\limits_{b_i^1\leq b_i^2} \sum\limits_{b_{-i}\in B^{n-1}} \sum\limits_{\theta_i\in \Theta_i} 
\left[(v_i-b_i^2)  \Pr(\mathcal{W}\vert b_i^2,\sigma_{-i},b_{-i}) \Pr(\sigma_i^2(v_i,b_i^1,\theta_i)=b_i^2\vert \theta_i, b_i^1)\Pr(\theta_i\vert b_i^1)\cdot \right. \nonumber \\
&\left. \cdot \Pr(\mathcal{S}\vert b_i^1,b_{-i})
\Pr(\sigma_i^1(v_i)=b_i^1)\Pr(\sigma_{-i}=b_{-i})\right] \nonumber
\end{eqnarray}
where $\mathcal{S}$ is the event bidder $i$ reaches the second stage and $\mathcal{W}$ is the event bidder $i$ wins in the second stage, and as proven above, the distribution of the signals is independent of the bids of the others.
The order of summation can be changed in the following manner:
\begin{eqnarray}
&\sum\limits_{b_i^2\in B} \sum\limits_{b_{-i}\in B^{n-1}} 
(v_i-b_i^2)  \Pr(\mathcal{W}\vert b_i^2,\sigma_{-i},b_{-i})\Pr(\sigma_{-i}=b_{-i})\cdot  \nonumber \\
&\cdot\left[\sum\limits_{b_i^1\leq b_i^2} \sum\limits_{\theta_i\in \Theta_i} \Pr(\mathcal{S}\vert b_i^1,b_{-i})
\Pr(\sigma_i^2(v_i,b_i^1,\theta_i)=b_i^2\vert \theta_i, b_i^1)\Pr(\theta_i\vert b_i^1)\Pr(\sigma_i^1(v_i)=b_i^1)\right] \nonumber
\end{eqnarray}
For each $b_{-i}$, the opponent of bidder $i$ in the second round is determined (maybe with some probability in case of a tie break) and along with $\sigma_{-i}$ also his second stage bid.
If bidder $i$ bids $b_i^2$ in the first round instead of $b_i^1$, he does not change the identity of his opponent and hence none of the terms in the outer summation changes.
On the other hand, $\Pr(\mathcal{S}\vert b_i^1,b_{-i})\leq \Pr(\mathcal{S}\vert b_i^2,b_{-i})$ so the inner summation in Eq.\eqref{eq:totalprob_b_i^2} is smaller than
\begin{equation}
\sum\limits_{\substack{ b_i^1\leq b_i^2 \\ \theta_i\in \Theta_i}} \Pr(\mathcal{S}\vert b_i^2,b_{-i})
\Pr(\sigma_i^2(v_i,b_i^1,\theta_i)=b_i^2\vert \theta_i, b_i^1)\Pr(\theta_i\vert b_i^1)\Pr(\sigma_i^1(v_i)=b_i^1) = \Pr(\mathcal{S}\vert b_i^2,b_{-i})\Pr(\hat{\sigma}_i^1(v_i)=b_i^2). \nonumber
\end{equation} 
Therefore, the expected payoff is smaller than
\begin{equation}
\sum\limits_{b_i^2\in B} \sum\limits_{b_{-i}\in B^{n-1}} 
(v_i-b_i^2)  \Pr(\mathcal{W}\vert b_i^2,\sigma_{-i},b_{-i})\Pr(\sigma_{-i}=b_{-i})\Pr(\hat{\sigma}_i^1(v_i)=b_i^2)  \nonumber
\end{equation}
which is exactly the expected payoff when using $(\hat{\sigma}_i^1,\hat{\sigma}_i^2)$ and the proof is complete.\hfill\end{proof}



A direct result of this is that in a non-informative information structure, all strategies that change the bid in the second stage are weakly dominated by strategies that do not.
If we limit our discussion to only non dominated strategies, all bidders use the same bid in the first and second stage, hence only the first stage matters.
As a result, a non-informative information structure can be analyzed as a first-price single-stage auction.
This simplifies the analysis in the examples to follow.

\section{Examples}\label{example}

We present two simple examples to illustrate the model and provide key insights into our results.
Specifically, the examples demonstrate how the revenue-maximizing information structure depends on the number of bidders.

The first example (Section \ref{sect:smallnumberofbidders}) demonstrates that when different information structures are compared, the results vary with the number of bidders. 
We show that with a small number of bidders, the optimal information structure for the auctioneer depends on the exact distribution of valuations.
In particular, an informative information structure can be better. 
However, with a larger number of bidders a non-informative information structure becomes optimal. 
The next question we address is how large the number of bidders needs to be to ensure that the non-informative structure is optimal.
The second example (Section \ref{Sect:largeN}) suggests that this number is not necessarily very large.
Moreover, we show that for a large enough number of bidders, both information structures are equivalent: under both structures, the only equilibrium is the revenue-maximizing strategy profile.
All the results are formally presented and proved in Section \ref{results}.

\subsection{Small Number of Bidders} \label{sect:smallnumberofbidders}


Suppose there are $n=3$ bidders, two valuations $\{v^L=0.01,v^H=1\}$, and four possible bids $B=\{0,0.33,0.66,0.99\}$.
The probability of each bidder having a high valuation is $p=0.01$, independent of the other bidders.
Clearly, a bidder with a low valuation bids $b^1=0$ in all equilibria.

\begin{figure}[ht]
\centering
\begin{tikzpicture}[x=1.5cm]
\draw[black,-,thick,>=latex]
  (0,0) node[below left] {$b^1=0$}  -- (10,0) node[below right] {$v^H=1$};
\foreach \Xc in {0,1,3,6,9,10}
{
  \draw[black,thick]
    (\Xc,0) -- ++(0,7pt) ;
}

  
  \node[below,align=left,anchor=north,inner xsep=0pt]
  at (1,0)
  {$v^L=0.01$};

  \node[below,align=left,anchor=north,inner xsep=0pt]
  at (3,0)
  {$b^2=0.33$};

  \node[below,align=left,anchor=north,inner xsep=0pt]
  at (6,0)
  {$b^3=0.66$};

  \node[below,align=left,anchor=north,inner xsep=0pt]
  at (9,0)
  {$b^4=0.99$};

\end{tikzpicture}
\end{figure}

We compare the two extreme information structures:
a non-informative information structure and a fully informative  information structure.
Claim~\ref{clm:noinfo} shows that the non-informative structure is not revenue-maximizing: there is no symmetric equilibrium where a bidder with a high valuation bids higher than $b^2=0.33$ with a positive probability.
The unique symmetric equilibrium, in this case, is the one where a high valuation bidder bids $b^2=0.33$ in both stages.

\begin{claim}\label{clm:noinfo}
In a non-informative information structure, the only symmetric equilibrium in non weakly dominated strategies is that where a high valuation bidder submits the bid $b^2$ in both stages.
\end{claim} 
\begin{proof} 
Fix an equilibrium. 
As mentioned above, it is easy to see that in this equilibrium, all low type bidders must bid $0$ in both stages.
Now consider a bidder with a high valuation.
When bidding $0.99$ at any stage, his profit cannot exceed $0.01$.
When bidding $0.66$ at any stage, his profit cannot exceed $0.34$.
When bidding $0$ in both stages, his profit cannot exceed $\tfrac{1}{3}$.
However, when bidding $0.33$ in both stages, both opponents have low valuations with probability $0.99^2$, so the profit is at least $0.99^2\cdot 0.67$, which is much higher. 
Finally, when bidding $0$ in the first stage and $0.33$ in the second stage, his profit is lower than when bidding $0.33$ in both stages (same profit if selected, but lower probability of being selected, as discussed in Lemma~\ref{thmNoInfoIsFPSA}).
Therefore, a bidder with a high valuation must bid $b^2= 0.33$ in both stages.
Conversely, it is easily seen that this results is an equilibrium.
\hfill\end{proof}

In the fully informative case, all equilibria are such that a high-type bidder never bids $b_1$ (because this would imply a gain less than $\tfrac{1}{3}$ in that case), and bids $0.66$ or more with positive probability in the second round (if the second round bid was always $b^2$ for all high type bidders, then deviating and bidding $b^3$ in the second round when the other bidder has the valuation $v^H$ is a profitable deviation).
Hence, any equilibrium of the fully informative model (there are several) is better for the seller than the best equilibrium without information.
Note also that the best equilibrium for the seller in the full information case is the one where, when there are at least two bidders with high type, they both bid $b^4$ in the second round.   

In this example, revealing information benefits the seller.
If there are only low valuation bidders, the information structures produce the same profit. 
However, given there is at least one high-type bidder, the profit of the seller in equilibrium without information is $0.33$, while the profit of the seller in the fully informative structure is higher.
For example, if in equilibrium a high-typed bidder bids $0.33$ and increases his bid to $0.66$ only if there is another high-typed bidder, the payoff to the seller would be $0.33\cdot 0.99^2 + 0.66 (1-0.99^2)>0.33$. 

Note that changes to the joint distribution of the values affect this observation. 
For example, if the dependency is such that a bidder with a high valuation induces a very high probability of an opponent having a high valuation as well, the information becomes irrelevant to determine the types of the others and can only serve as a coordination device.

Direct computation reveals that the reasoning in the fully-informative model remains true as long as $n\leq 184$ (for $n>184$ the two information structures are revenue-equivalent), while Claim \ref{clm:noinfo} remains true only as long as $n\leq 155$.
When the number of bidders is between $155$ and $184$, the only equilibrium without information is to bid $b^3$ in both stages, while the equilibrium in the fully-informed model is to bid $b^2$ in the first stage and to raise the bid only if necessary.
Hence, for the auctioneer, not releasing information becomes more profitable  for a large number of bidders.\footnote{For a much larger number of bidders ($\sim 600$), the only equilibrium is the revenue-maximizing equilibrium, as Theorem \ref{thmEffUniqueWithLargeN} suggests.}
An auction with more than $100$ bidders is atypical.
However, the next example demonstrates that revealing no information may prove advantageous even with a much smaller number of bidders, depending on the other auction parameters.

\subsection{Large Number of Bidders}\label{Sect:largeN}
Suppose that there are $n>3$ bidders, two private values, $V=\{v^L=0.3, v^H=1\}$, and five possible equally-spaced bids $B=\{0,0.2,\ldots,0.8\}$.
Each bidder has a probability $p$ of having the private value $v^H$, independent of the other bidders. 
In $\Sigma^*$, a bidder with private value $v^L$ bids $0.2$  and a bidder with private value $v^H$ bids $0.8$.
If everyone is bidding according to $\Sigma^*$, there is no profitable deviation for a bidder with private value $v^L$ regardless of the information structure, so we consider only a high-type bidder.

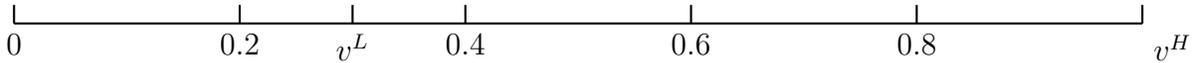
\begin{figure}[ht]
\centering
\begin{tikzpicture}[x=1.5cm]
\draw[black,-,thick,>=latex]
  (0,0) -- (10,0) node[below right] {$v^H$};
\foreach \Xc in {0,2,3,4,6,8,10}
{
  \draw[black,thick]
    (\Xc,0) -- ++(0,7pt) ;
}

  \node[below,align=left,anchor=north,inner xsep=0pt]
  at (0,0)
  {$0$};

  \node[below,align=left,anchor=north,inner xsep=0pt]
  at (2,0)
  {$0.2$};
  
  \node[below,align=left,anchor=north,inner xsep=0pt]
  at (3,0)
  {$v^L$};

  \node[below,align=left,anchor=north,inner xsep=0pt]
  at (4,0)
  {$0.4$};

  \node[below,align=left,anchor=north,inner xsep=0pt]
  at (6,0)
  {$0.6$};

  \node[below,align=left,anchor=north,inner xsep=0pt]
  at (8,0)
  {$0.8$};

\end{tikzpicture}
\end{figure}

In the non-informative information structure, $\Sigma^*$ is an equilibrium iff a high-type bidder has no profitable deviation.
Of all possible deviations, the most profitable is where he bids the lowest bid that is above $v^L$ in both stages.
The expected profit from deviating is the gain $(1-0.4)$ times the probability that there are no high type bidder who outbid him.
The expected equilibrium payoff is the gain from winning the auction $(1-0.8)$ times the probability of being chosen among all high type bidders (who also bid $0.8$) which can be found as we later see using Eq.~\eqref{eq expectation}.
The condition for $\Sigma^*$ to be an equilibrium is therefore
\begin{equation}
(1-0.8)\tfrac{1-(1-p)^n}{np} \geq (1-0.4)(1-p)^{n-1}. \nonumber
\end{equation}
For each $p$, there exists a minimal $n$ for which this inequality holds, denoted by $N^{NI}$, and it remains true for every $n>N^{NI}$.
This $N^{NI}$ is the minimum number of bidders required for $\Sigma^*$ to be an equilibrium.

Naturally, $N^{NI}$ decreases with $p$.
When $p$ is small, a high-type bidder can bid the lower bid because there is a high probability that he is the only high-type bidder and no one will outbid him.
When $p$ is large, there is a high probability that other high-type bidders exist and bid $0.8$, and a lower bid is bound to lose.

A similar computation can be done for any information structure.
For example, in the fully-informative structure, the most profitable deviation from $\Sigma^*$ for a high-type bidder is to bid $0.4$ in the first stage and raise his bid in the BAFO stage only if the other bidder outbid him.
The condition for $\Sigma^*$ to be an equilibrium is
\begin{equation}
(1-0.8)\tfrac{1-(1-p)^n}{np} \geq (1-0.4)(1-p)^{n-1}+\tfrac{1-0.8}{2}(n-1)p(1-p)^{n-2}. \nonumber
\end{equation}
Observe that to satisfy this inequality a larger $n$ is needed compared to the former one.
Again, this inequality can be solved for each $p$ to obtain $N^{I}$, the minimum number of bidders required for $\Sigma^*$ to be an equilibrium in this information structure.

The general result is proved in Theorem \ref{thmEffUniqueWithLargeN} -- for $n$ large enough, $\Sigma^*$ is the unique equilibrium  regardless of the information structure.
The values of $N^{NI}$ and $N^{I}$ as a function of $p$ are summarized in Figure \ref{fig:minimalNumofBidders}.
As Corollary \ref{thm:withoutinfo less bidders} suggests, $N^{I}$ is always larger than $N^{NI}$. 
When the probability of being a high-type bidder is not extremely small, the required number of bidders is around $10$ for both types of information structures, and when the probability of being a high-type bidder is around $0.5$, it is around $5$.

\begin{figure}
\begin{center}
\begin{tikzpicture}

\begin{axis}[axis x line=middle, axis y line=middle, xtick={0.1,0.2,...,0.5},ytick={3,6,...,30}, xmin=0.05, xmax=0.6, ymin=3, ymax=30,xlabel = {Probability of having type $v^H$}, x label style={at={(axis description cs:0.5,-0.1)},anchor=north}]
    \addplot[] coordinates {
(0.05,38) (0.06,38) (0.06,32) (0.07,32) (0.07,28) (0.08,28) (0.08,24) (0.09,24) (0.09,21) (0.10,21) (0.10,19) (0.11,19) (0.11,18) (0.12,18) (0.12,16) (0.13,16) (0.13,15) (0.14,15) (0.14,14) (0.15,14) (0.15,13) (0.16,13) (0.16,12) (0.17,12) (0.17,11) (0.18,11) (0.18,11) (0.19,11) (0.19,10) (0.20,10) (0.20,10) (0.21,10) (0.21,9) (0.22,9) (0.22,9) (0.23,9) (0.23,9) (0.24,9) (0.24,8) (0.25,8) (0.25,8) (0.26,8) (0.26,8) (0.27,8) (0.27,7) (0.28,7) (0.28,7) (0.29,7) (0.29,7) (0.30,7) (0.30,7) (0.31,7) (0.31,6) (0.32,6) (0.32,6) (0.33,6) (0.33,6) (0.34,6) (0.34,6) (0.35,6) (0.35,6) (0.36,6) (0.36,6) (0.37,6) (0.37,5) (0.38,5) (0.38,5) (0.39,5) (0.39,5) (0.40,5) (0.40,5) (0.41,5) (0.41,5) (0.42,5) (0.42,5) (0.43,5) (0.43,5) (0.44,5) (0.44,5) (0.45,5) (0.45,4) (0.46,4) (0.46,4) (0.47,4) (0.47,4) (0.48,4) (0.48,4) (0.49,4) (0.49,4) (0.50,4) (0.50,4) 
};
    \node at (axis cs: 0.1,12) {$N^{NI}$};

    \addplot[] coordinates {
(0.06,40) (0.07,40) (0.07,35) (0.08,35) (0.08,30) (0.09,30) (0.09,27) (0.10,27) (0.10,24) (0.11,24) (0.11,22) (0.12,22) (0.12,20) (0.13,20) (0.13,19) (0.14,19) (0.14,17) (0.15,17) (0.15,16) (0.16,16) (0.16,15) (0.17,15) (0.17,14) (0.18,14) (0.18,14) (0.19,14) (0.19,13) (0.20,13) (0.20,12) (0.21,12) (0.21,12) (0.22,12) (0.22,11) (0.23,11) (0.23,11) (0.24,11) (0.24,10) (0.25,10) (0.25,10) (0.26,10) (0.26,9) (0.27,9) (0.27,9) (0.28,9) (0.28,9) (0.29,9) (0.29,8) (0.30,8) (0.30,8) (0.31,8) (0.31,8) (0.32,8) (0.32,8) (0.33,8) (0.33,7) (0.34,7) (0.34,7) (0.35,7) (0.35,7) (0.36,7) (0.36,7) (0.37,7) (0.37,7) (0.38,7) (0.38,7) (0.39,7) (0.39,6) (0.40,6) (0.40,6) (0.41,6) (0.41,6) (0.42,6) (0.42,6) (0.43,6) (0.43,6) (0.44,6) (0.44,6) (0.45,6) (0.45,6) (0.46,6) (0.46,5) (0.47,5) (0.47,5) (0.48,5) (0.48,5) (0.49,5) (0.49,5) (0.50,5) (0.50,5) 
};

\node at (axis cs: 0.3,12) {$N^I$};
\end{axis}
\end{tikzpicture}
\end{center}
\caption{The minimum number of bidders for $\Sigma^*$ to become an equilibrium when no information is given before the BAFO stage ($N^{NI}$) and when all the bids are revealed before the BAFO stage ($N^I$) for the auction presented in Section \ref{Sect:largeN}. For larger probabilities, both functions continue to decrease towards the minimal possible number of bidders in this model, $n=3$.}.
\label{fig:minimalNumofBidders}
\end{figure}
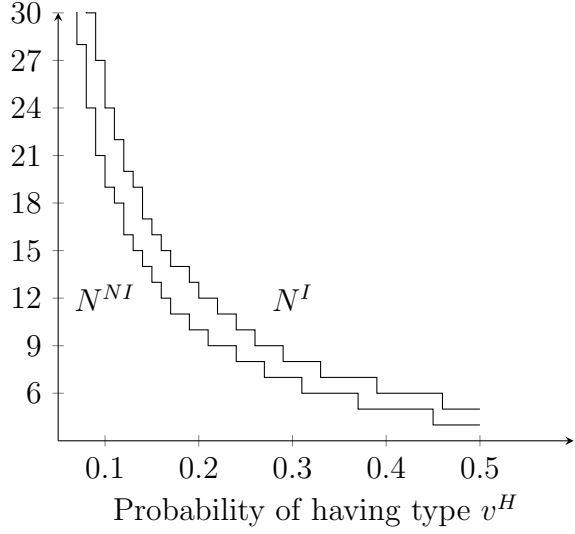

\section{The Main Result}\label{results}

Our main result is that the revenue-maximizing strategy profile is the unique equilibrium for a large enough number of bidders, and that the required number of bidder increases with the information (in the sense of Definition \ref{def:more_inf}).
The intuition is that as the number of bidders increases, the probability of being chosen for the second stage, and ultimately winning, decreases.
Hence, for large enough number of bidders and regardless of the information, it is optimal to bid as high as possible to increase the chances of being selected.

\begin{theorem}\label{thmEffUniqueWithLargeN}
For any set of private values $V$, information structure $\Theta$, and family of distributions over $V$ which satisfy the uniform full support assumption $(F_n)_{n\in \N}$ (with a lower bound $\d>0$), there exists $N_0$ such that if the number of bidders is larger than $N_0$, then the unique symmetric equilibrium is $\Sigma^*$.  
\end{theorem}
\begin{proof}
Assume, by contradiction, that the result does not hold.
Then the set of $n\in\N$ for which $\Sigma^*$ is not the unique symmetric equilibrium is infinite.
For each $n$ in this set (actually, a sequence), there exists an equilibrium in which there exists a type $v^j$ that bids $\beta^j$ in the first stage with probability less than $1$.
Since the number of types is finite and the number of bids is finite, there exists a sub-sequence of these $n$s in which in equilibrium the highest type that does not play his respective $\beta$ in the first stage with probability $1$ is $v^k$ for some $k$, and the lowest bid he bids in equilibrium with nonzero probability is the same, $b_{\min}<\beta^k$.
W.l.o.g. we can assume that $k=K$ and this bidder is the bidder with the highest possible valuations.
Otherwise, in equilibrium, all bidders with types $v^j>v^k$ bid their respective $\beta^j$ and outbid him.
His only chance for a positive profit is if there are no bidders with a higher valuation.
We can denote this event by $D$ and the rest of the proof is identical even if $k<K$, when everything is conditional on $D$.

Denote by $q_2(n)$ the probability that a bidder with valuation $v^K$ will bid $b_{\min}$ and let  $q^{lim}_2=\lim_{n\rightarrow\infty}q_2(n)$ (the limit exists up to considering a subsequence).
Also denote $q_1(n) = 1- q_2(n)$.
Let $A$ be the expected payoff of the bidder with valuation $v^K$ in the equilibrium in which he is bidding $b_{\min}$ at stage 1 and let $B$ be the payoff when he deviates to bid $\beta^K$.
We divide the remaining discussion into two cases: $q_2^{lim}<1$ and $q_2^{lim}=1$.

\textbf{Case 1 : $q_2^{lim}<1$.} 

In that case, for $n$ large enough, there is a probability bounded away from zero that a bidder with valuation $v^k$  will bid higher than $b_{min}$.
Formally, there exists $\bar{q}\in(0,1)$ and $N_0\in\mathbb{N}$ such that for all $n\geq N_0$:  $q_1(n)\geq \bar{q}>0$. 
In both stages, no bidder bids strictly more than $\beta^K$.
Thus, the probability of a bidder with valuation $v^K$ winning the auction when using strategy $\sigma^*$ can be (very loosely) bounded from below by $\frac{1}{n}$.
To conclude, $B\geq (v^K-\beta^K)\frac{1}{n}$.

When a bidder bids $b_{min}$ in the first stage, he will move to the second stage (and thus have some positive probability of winning) only if at most one bidder bids $\beta^K$ in the first stage.
Therefore, 
\begin{equation}
A\leq \left[(1-\d\bar{q})^{n-1}+(n-1)(1-\d\bar{q})^{n-2}\d\bar{q} \right](v^K-b_{min})
. \nonumber
\end{equation}

For large enough $n$, this bound converges to $0$ exponentially while the lower bound on $B$ converges to $0$ more slowly.
Thus, for large enough $n$, $A<B$.

\textbf{Case 2 : $q_2^{lim}=1$.} 

In this case, at the limit, no bidder bids $\beta^K$ in the first stage.
Formally, for every $\varepsilon>0$ there exists $N_0\in\mathbb{N}$ such that the probability that a bidder with valuation $v^K$ bids $\beta^K$ in the first stage is smaller than $\varepsilon$. 
If a bidder with valuation $v^K$ bids according to $\sigma^*$, then he wins with the probability of being chosen among all those bidding the same. 

The number of bidders who have high valuation and bid $\beta^K$ can be bounded from above using a binomial random variable with the parameters $n$ and $\varepsilon$.
It is well-known that for a binomial random variable $X\sim Bin(n,p)$, 
\begin{equation}\label{eq expectation}
\mathbb{E}\left[\tfrac{1}{X+1}\right]=\tfrac{1}{np}\left[1-(1-p)^n\right],
\end{equation}
therefore, the following bound is obtained:
For $\varepsilon$ small enough and $n$ large enough $B\geq \tfrac{v^K-\beta^K}{n \varepsilon}(1-(1-\varepsilon)^n)=f_\varepsilon(n)$.

A bidder who bids $b_{min}$ can win only if there is at most one bidder who bids above $b_{min}$ in the first stage.
Denote this event by $E$.
Given $E$, the bidder needs to be chosen among all the other bidders who bid $b_{min}$.

Conditioning on $E$, the probability to win is lower than if all other bidders were (independently) bidding $b_{min}$ with probability $\delta(1-\varepsilon)$ and strictly lower with the remaining probability.

From Eq.~\eqref{eq expectation} again, we have
\begin{equation}
A \leq \Pr(E) \tfrac{v^K-b_{min}}{n\d(1-\varepsilon)}\leq 2\tfrac{v^K-b_{min}}{n\d}=g_\varepsilon(n). \nonumber
\end{equation}

Since $\tfrac{f_\varepsilon(n)}{g_\varepsilon(n)}\to O(\tfrac{1}{\epsilon})$ as $n\to\infty$,
then for $\varepsilon$ small enough there exists $N_0$ such that for every $n>N_0$, $A<B$.
To conclude, $b_{min}$  cannot be played in the first stage in equilibrium.
\hfill\end{proof}

Our main goal is to compare information structures.
More precisely, we wish to compare the expected revenue of the seller when using two information structure that can be pair-wise compared (and in particular, the non-informative information structure and any informative one).
Clearly, the revenue-maximizing equilibrium provides the best expected equilibrium payoff for the seller.
She therefore strives to design an information structure under which the revenue-maximizing equilibrium will be the unique one.
Our next theorem states that whenever the revenue-maximizing equilibrium appears for some structure $\Theta$, it also appears for structures less informative than $\Theta$.

\begin{theorem}\label{thmEFFEqWithIsAlsoWithout}
When $\Sigma^*$ is an equilibrium in an auction with some information structure $\Theta$, it is also an equilibrium in an auction with any information structures $\Theta'$ which is less informative than $\Theta$.

Thus, whenever $\Sigma^*$ is an equilibrium in an auction with some informative information structure, it is also an equilibrium in an auction with any non-informative information structure.
\end{theorem}
\begin{proof}
Suppose that $\Sigma^*$ is an equilibrium in an auction with the information structure $\Theta$ and assume by contradiction that bidder $i$ has a profitable deviation from the strategy profile $\Sigma^*$ in an auction with the less informative structure $\Theta'$.
Denote this strategy by $\hat{\sigma}_i=(\hat{\sigma}_i^1,\hat{\sigma}_i^2)$ and let $L_i$ denote the garbling function that corresponds to the strategy profile $\Sigma^*_{-i}$.

Assume that the information structure is $\Theta$ and consider the following strategy.
At the first stage, bid according to $\hat{\sigma}_i^1$ and denote the bid by $b_i^1$.
In the second stage, after receiving the signal $\theta$, use $L_i(\theta,b_i^1)$ to choose a new signal at random from $\Theta$.
Denote this signal by $\theta'$.
Finally, bid as dictated by $\hat{\sigma}_i^2(v_i,b_i^1,\theta')$.
By using this strategy, bidder $i$ disregards the information provided by $\Theta$ and acts as if the information structure is actually $\Theta'$.
Since $\hat{\sigma}$ is a profitable deviation with $\Theta'$, this is also a profitable deviation with $\Theta$, in contradiction to the fact that $\Sigma^*$ is an equilibrium with the information structure $\Theta$.
\hfill\end{proof}

This result can easily be generalized. 
A strategy profile that is an equilibrium when information is available but not used, is still an equilibrium when this information is not available, simply because the set of available deviations is smaller without information.
For example, if bidding $b^1$ regardless of type is an equilibrium in some informative information structure it is also an equilibrium in non-informative information structures, as the information about the bids does not change the prior regarding the bidders' types.

A direct corollary of this result is that a smaller number of bidders is needed to generate the revenue-maximizing equilibrium when no information is conveyed to the bidders.
Any information provided by the seller might serve as a basis for profitable deviations.

\begin{corollary}\label{thm:withoutinfo less bidders}
Let $\Theta$ be an information structure and denote by $N^{\Theta}$ the minimal number of bidders such that whenever there are more bidders, $\Sigma^*$ is an equilibrium with the information structure $\Theta$.
If $\Theta'$ is less informative than $\Theta$, then $N^{\Theta}\geq N^{\Theta'}$.

In particular, if $NI$ is some non-informative information structure and $FI$ is a fully-informative information structure, then $N^{NI}\leq N^{\Theta}\leq N^{FI}$. 
\end{corollary}
\begin{proof}
For $N^{\Theta}$ bidders, $\Sigma^*$ is an equilibrium when the information is $\Theta$.
According to  Theorem \ref{thmEFFEqWithIsAlsoWithout}, $\Sigma^*$ is also an equilibrium with the less informative information structure $\Theta'$ with $N^{\Theta}$ bidders.
Since $N^{\Theta'}$ is the minimum number of bidders for which $\Sigma^*$ is an equilibrium with the information structure $\Theta'$, it follows that $N^{\Theta'}\leq N^{\Theta}$.\hfill\end{proof}

The required number of bidders can be computed in each auction based on the auction parameters  (bids, private values, and $F_n$).
We demonstrate this calculation in Section \ref{Sect:largeN}.
It is clear from an examination of the mechanism behind the example that for reasonable parameter values, ten or even fewer bidders could be sufficient for the revenue-maximizing equilibrium to emerge. 

\section{Concluding Remarks}\label{Sect:Discussion}

This paper sought to determine the optimal information about bids that a seller should disclose to the bidders so as to increase her revenue in two-stage auctions.
The answer to this question strongly depends on the number of bidders.
When the number of bidders is small, all the parameters of the auction need to be considered;  there are examples where it is better to reveal information and others where it is better to withhold it.
For a very large number of bidders, the answer is that the information structure is irrelevant, as the only equilibrium is the one where the bidders submit their maximum bid, regardless of the information.
Moreover, this strategy profile also becomes an equilibrium for fewer bidders when no information is revealed, than under any informative information structure.

We showed that the revenue-maximizing strategy profile requires fewer bidders to be an equilibrium when less information is provided.
We conjecture that this also applies to uniqueness -- fewer bidders are required for the revenue-maximizing equilibrium to be the \emph{unique} equilibrium when the information structure is non-informative.
Any information revealed by the seller can create additional equilibria.
This strengthens the idea that the best tactic on information is to reveal nothing -- not only does the revenue-maximizing profile is an equilibrium first in this case, it also becomes the unique equilibrium first.

As elsewhere (e.g. \cite{quint2018theory}), uniqueness can be shown only in special cases.
We were able to show that this holds for a simple example with pure strategies, which can be generalized to a large class of two-stage auctions.
This proof is given in Appendix \ref{app:uniquness}.
In addition, our conjecture regarding uniqueness was verified by computer simulation on another class of auctions for a large set of parameters.

We conclude that when the number of bidders is not too small\footnote{As a rule of thumb, the number of bidders required is such that for every private valuation $v^j$, the expected number of bidders that have this private value given it is the maximal private value anyone has, is at least $3$.
For example, if valuations are independent and each bidder has equal probability to have any of the private valuation, there should be roughly $n=3K$ bidders for a bidder with the valuation $v^K$ to feel enough competitive pressure to raise his bids.}  or is unknown, there is some advantage to not revealing information.
This might allow the revenue-maximizing strategy profile to become an equilibrium and possibly the unique equilibrium.

If the sole purpose of the second stage is to encourage competition and the valuations are private, then it is advisable to consider simplifying the process and conducting a single-stage sealed-bid auction.
However, the main purpose of the second stage may be to discuss design (in design-build auctions) or to enable the bidders to learn more about the item being auctioned, either from the auctioneer or from the bids of the other bidders (when valuations have a common-value component).
In such cases, the optimal information policy  regarding first-stage bids remains to be determined by future research. 

\bibliographystyle{plainnat}
\bibliography{CitationsTwoStageAuctions}

\begin{thebibliography}{34}
\providecommand{\natexlab}[1]{#1}
\providecommand{\url}[1]{\texttt{#1}}
\expandafter\ifx\csname urlstyle\endcsname\relax
  \providecommand{\doi}[1]{doi: #1}\else
  \providecommand{\doi}{doi: \begingroup \urlstyle{rm}\Url}\fi

\bibitem[Ahadzi and Bowles(2001)]{ahadzi2001private}
Marcus Ahadzi and Graeme Bowles.
\newblock The private finance initiative: the procurement process in
  perspective.
\newblock In \emph{17th annual ARCOM Conference}, pages 5--7, 2001.

\bibitem[Ashkenazi-Golan and Lehrer(2019)]{ashkenazi2019blackwell}
Galit Ashkenazi-Golan and Ehud Lehrer.
\newblock Blackwell's comparison of experiments and discounted repeated games.
\newblock \emph{Games and Economic Behavior}, 117:\penalty0 163--194, 2019.

\bibitem[Bali and Jackson(2002)]{bali2002asymptotic}
Valentina Bali and Matthew Jackson.
\newblock Asymptotic revenue equivalence in auctions.
\newblock \emph{Journal of Economic Theory}, 106\penalty0 (1):\penalty0
  161--176, 2002.

\bibitem[Blackwell(1951)]{blackwell1951comparison}
David Blackwell.
\newblock Comparison of experiments.
\newblock In \emph{Proceedings of the Second Berkeley Symposium on Mathematical
  Statistics and Probability}, pages 93--102. Berkeley and Los Angeles,
  University of California Press, 1951.

\bibitem[Boone and Goeree(2009)]{boone2009optimal}
Jan Boone and Jacob~K Goeree.
\newblock Optimal privatisation using qualifying auctions.
\newblock \emph{The Economic Journal}, 119\penalty0 (534):\penalty0 277--297,
  2009.

\bibitem[Bulow and Klemperer(1994)]{bulow1994auctions}
Jeremy Bulow and Paul Klemperer.
\newblock Auctions vs. negotiations.
\newblock Technical report, National Bureau of Economic Research, 1994.

\bibitem[Caffarelli(1998)]{caffarelli1998auction}
Filippo~Vergara Caffarelli.
\newblock An auction mechanism for privatisation: The eni group multiple round
  auction procedure.
\newblock \emph{Rivista Di Politica Economica}, 88:\penalty0 99--128, 1998.

\bibitem[Chelekis(1992)]{chelekis1992official}
George~C Chelekis.
\newblock \emph{The Official Government Auction Guide}.
\newblock Crown, 1992.

\bibitem[Dudkin and V{\"a}lil{\"a}(2006)]{dudkin2006transaction}
Gerti Dudkin and Timo V{\"a}lil{\"a}.
\newblock Transaction costs in public-private partnerships: a first look at the
  evidence.
\newblock \emph{Competition and regulation in network industries}, 1\penalty0
  (2):\penalty0 307--330, 2006.

\bibitem[Dutra and Menezes(2002)]{dutra2002hybrid}
Joisa~C Dutra and Flavio~M Menezes.
\newblock Hybrid auctions.
\newblock \emph{Economics Letters}, 77\penalty0 (3):\penalty0 301--307, 2002.

\bibitem[Fibich and Gavious(2010)]{fibich2010large}
Gadi Fibich and Arieh Gavious.
\newblock Large auctions with risk-averse bidders.
\newblock \emph{International Journal of Game Theory}, 39\penalty0
  (3):\penalty0 359--390, 2010.

\bibitem[Finley(2001)]{finley2001dynamic}
Kenneth Finley.
\newblock Dynamic pricing possibilities in the purchase of bulk fuel for the
  department of defense.
\newblock Technical report, Naval Postgraduate School Monterey CA, 2001.

\bibitem[Ganuza(2007)]{ganuza2007competition}
Juan-Jos{\'e} Ganuza.
\newblock Competition and cost overruns in procurement.
\newblock \emph{The Journal of Industrial Economics}, 55\penalty0 (4):\penalty0
  633--660, 2007.

\bibitem[Hernando-Veciana and Michelucci(2018)]{hernando2018inefficient}
{\'A}ngel Hernando-Veciana and Fabio Michelucci.
\newblock Inefficient rushes in auctions.
\newblock \emph{Theoretical Economics}, 13\penalty0 (1):\penalty0 273--306,
  2018.

\bibitem[Hillas and Liu(2016)]{hillas2016correlated}
John Hillas and Min Liu.
\newblock Correlated equilibria of two person repeated games with random
  signals.
\newblock \emph{International Journal of Game Theory}, 45\penalty0
  (1-2):\penalty0 137--153, 2016.

\bibitem[Lehrer(1992)]{lehrer1992equilibrium}
Ehud Lehrer.
\newblock On the equilibrium payoffs set of two player repeated games with
  imperfect monitoring.
\newblock \emph{International Journal of Game Theory}, 20\penalty0
  (3):\penalty0 211--226, 1992.

\bibitem[Merrifield et~al.(2002)Merrifield, Manchidi, and
  Allen]{merrifield2002asset}
Andrew Merrifield, Tjiamogale~Eric Manchidi, and Stephen Allen.
\newblock The asset procurement and operating partnership system (apops) for
  prisons in south africa.
\newblock \emph{International Journal of Project Management}, 20\penalty0
  (8):\penalty0 575--582, 2002.

\bibitem[Milgrom(1979)]{milgrom1979convergence}
Paul~R. Milgrom.
\newblock A convergence theorem for competitive bidding with differential
  information.
\newblock \emph{Econometrica}, 47\penalty0 (3):\penalty0 679--688, 1979.

\bibitem[Milgrom(1985)]{milgrom1985economics}
Paul~R Milgrom.
\newblock The economics of competitive bidding: a selective survey.
\newblock \emph{Social goals and social organization}, pages 261--289, 1985.

\bibitem[Milgrom and Weber(1982)]{milgrom1982theory}
Paul~R. Milgrom and Robert~J. Weber.
\newblock A theory of auctions and competitive bidding.
\newblock \emph{Econometrica}, 50\penalty0 (5):\penalty0 1089--1122, 1982.

\bibitem[Myerson(1981)]{myerson1981optimal}
Roger~B Myerson.
\newblock Optimal auction design.
\newblock \emph{Mathematics of operations research}, 6\penalty0 (1):\penalty0
  58--73, 1981.

\bibitem[Neeman(2003)]{neeman2003effectiveness}
Zvika Neeman.
\newblock The effectiveness of english auctions.
\newblock \emph{Games and Economic Behavior}, 43\penalty0 (2):\penalty0
  214--238, 2003.

\bibitem[Nigerian-Communications-Commission(2007)]{ncc2007}
Nigerian-Communications-Commission.
\newblock Information memorandum 800 mhz spectrum auction.
\newblock Technical report, Nigerian Communications Commission, 2007.
\newblock URL
  \url{http://www.ncc.gov.ng/index.php?option=com_docman&task=doc_download&gid=165&Itemid=}.

\bibitem[Noumba and Dinghem(2005)]{noumba2005private}
Paul Noumba and Severine Dinghem.
\newblock Private participation in infrastructure projects in the republic of
  {K}orea.
\newblock \emph{World Bank Policy Research Working Paper}, \penalty0 (3689),
  2005.

\bibitem[Perry et~al.(2000)Perry, Wolfstetter, and Zamir]{perry2000sealed}
Motty Perry, Elmar Wolfstetter, and Shmuel Zamir.
\newblock A sealed-bid auction that matches the english auction.
\newblock \emph{Games and Economic Behavior}, 33\penalty0 (2):\penalty0
  265--273, 2000.

\bibitem[Pesendorfer and Swinkels(2000)]{pesendorfer2000efficiency}
Wolfgang Pesendorfer and Jeroen~M Swinkels.
\newblock Efficiency and information aggregation in auctions.
\newblock \emph{American Economic Review}, 90\penalty0 (3):\penalty0 499--525,
  2000.

\bibitem[Quint and Hendricks(2018)]{quint2018theory}
Daniel Quint and Kenneth Hendricks.
\newblock A theory of indicative bidding.
\newblock \emph{American Economic Journal: Microeconomics}, 10\penalty0
  (2):\penalty0 118--51, 2018.

\bibitem[Rintala et~al.(2008)Rintala, Root, Ive, and
  Bowen]{rintala2008organizing}
Kai Rintala, David Root, Graham Ive, and Paul Bowen.
\newblock {Organizing a bidding competition for a toll road concession in South
  Africa: the case of Chapma's peak drive}.
\newblock \emph{Journal of Management in Engineering}, 24\penalty0
  (3):\penalty0 146--155, 2008.

\bibitem[Roll(2000)]{roll2000online}
Bruce~A Roll.
\newblock Online reverse auctions: a pricing tool for government contracting.
\newblock Technical report, Naval Postgraduate School Monterey CA, 2000.

\bibitem[Sauvet-Goichon(2007)]{SAUVETGOICHON200775}
Bruno Sauvet-Goichon.
\newblock {Ashkelon desalination plant -- A successful challenge}.
\newblock \emph{Desalination}, 203\penalty0 (1):\penalty0 75 -- 81, 2007.

\bibitem[Swinkels(2001)]{swinkels2001efficiency}
Jeroen~M Swinkels.
\newblock Efficiency of large private value auctions.
\newblock \emph{Econometrica}, 69\penalty0 (1):\penalty0 37--68, 2001.

\bibitem[Wilson(1977)]{wilson1977bidding}
Robert Wilson.
\newblock A bidding model of perfect competition.
\newblock \emph{The Review of Economic Studies}, 44\penalty0 (3):\penalty0
  511--518, 1977.

\bibitem[Wilson(1985)]{wilson1985incentive}
Robert Wilson.
\newblock Incentive efficiency of double auctions.
\newblock \emph{Econometrica}, 53\penalty0 (5):\penalty0 1101--1115, 1985.

\bibitem[Ye(2007)]{ye2007indicative}
Lixin Ye.
\newblock Indicative bidding and a theory of two-stage auctions.
\newblock \emph{Games and Economic Behavior}, 58\penalty0 (1):\penalty0
  181--207, 2007.

\end{thebibliography}

\newpage
\appendix
\numberwithin{equation}{section}

\section{Uniqueness of the Revenue-Maximizing Equilibrium}\label{app:uniquness}

We show that $\Sigma^*$ is the unique equilibrium even for fewer bidders when the information structure is non-informative, taking the case of two possible valuations, three possible bids and considering pure strategies only.
We conjecture that this is the case in general, regardless of the specific parameters of the auction.
Partial results in the general case, based on computer simulations, are also reported.

The following proposition compares two extreme information structures:
the non-informative information structure and the fully-informative information structure.

\begin{proposition}\label{prop:uniqueness simple case}
Consider an auction with $n$ bidders, two possible types $V=\{v^1,v^2\}$, and three possible bids $B=\{b^1,b^2,b^3\}$ ordered in the following way: $b^1<v^1<b^2<b^3<v^2$.
Assume that the valuations are independent and identically distributed, with probability $p$ of each bidder having valuation $v^2$. 
If $\Sigma^*$ is the unique equilibrium in pure strategies when the information structure is fully informative, then it is also the unique equilibrium in pure strategies when the information structure is non-informative.
\end{proposition}

\begin{proof}
Assume w.l.o.g. that $b^1=0$ and $v^2=1$.
Since $b^1<v^1<b^2<b^3<v^2$, bidders of type $v^1$ always bid in equilibrium $b^1$.
Bidders of type $v^2$ can choose any of the three possible bids.
In the rest of the discussion, we focus on one such bidder.
We denote by $X$ the number of other high-type bidders, i.e. $X\sim Bin(n-1,p)$.
The possible bids, private values, and revenues are illustrated in this figure:
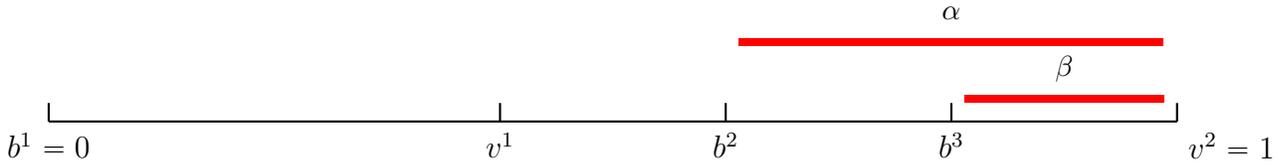
\begin{figure}[ht]
\centering
\begin{tikzpicture}[x=1.5cm]
\draw[black,-,thick,>=latex]
  (0,0) -- (10,0) node[below right] {$v^2=1$};
\foreach \Xc in {0,4,6,8,10}
{
  \draw[black,thick]
    (\Xc,0) -- ++(0,7pt) ;
}

  \fill[red]
    ([xshift=5pt]6,1)
      rectangle node[above, color=black] {\strut\small$\alpha$}
    ([xshift=-5pt]10,1.1);

 \fill[red]
    ([xshift=5pt]8,0.25)
      rectangle node[above, color=black] {\strut\small$\beta$}
    ([xshift=-5pt]10,0.35);

  \node[below,align=left,anchor=north,inner xsep=0pt]
  at (0,0)
  {$b^1=0$};
  
  \node[below,align=left,anchor=north,inner xsep=0pt]
  at (4,0)
  {$v^1$};

  \node[below,align=left,anchor=north,inner xsep=0pt]
  at (6,0)
  {$b^2$};

  \node[below,align=left,anchor=north,inner xsep=0pt]
  at (8,0)
  {$b^3$};
\end{tikzpicture}
\end{figure}

The strategy of the proof is to construct a counter-example -- a choice of parameters for which $\Sigma^*$ is the unique equilibrium in the informed model but not in the non-informed model.
We will show that such a counter-example is impossible.

\emph{The non-informed model}

For a counter-example, we need an equilibrium in pure strategies additional to $\Sigma^*$ in the non-informed model, but no additional equilibria in the fully-informed model.

Clearly, if the strategy profile in which everyone bids $b^1$ is an equilibrium in the non-informed model, then it is also an equilibrium in the fully-informed model.
This is because they are essentially the same -- no new information is revealed before the BAFO stage, as all bids are known and are the same regardless of private value.
Therefore in a counter-example, it cannot be an equilibrium, which means that there exists a profitable deviation from this strategy profile. 
The profitable deviation can be to bid $b^2$ or to bid $b^3$. 
However, if all other bidders bid $b^1$, bidding $b^3$ has no strategic advantage over bidding $b^2$: both guarantee a win, and the revenue is higher when $b^2$ is bid. 
Thus, bidding $b^2$ is a profitable deviation:

\begin{equation}
U_i(b^1\vert v^2)=\frac{1}{n}<\alpha=U_i(b^2,b^1_{-i}\vert v^2) \nonumber
\end{equation}

The only other pure strategy profile is to bid $b^2$ (denoted by $\Sigma$), so it is assumed to be the additional equilibrium.
Hence, bidding $b^3$ is not a profitable deviation:
\begin{equation}
U_i(\Sigma\vert v^2)=\alpha\E(\frac{1}{X+1})\geq U_i(b^3,\Sigma_{-i}\vert v^2)=\beta \nonumber
\end{equation}
In addition, $\Sigma^*$ is equilibrium in the non-informed model so bidding $b^2$ is not a profitable deviation:
\begin{equation}
U_i(\Sigma^*\vert v^2)=\beta\E(\frac{1}{X+1})\geq U_i(b^2,\Sigma_{-i}^*\vert v^2)=\alpha \Pr(X=0) \nonumber
\end{equation}

\emph{The fully-informed model}

We construct an example where $\Sigma^*$ is the unique equilibrium with information.
Thus, there is no profitable deviation from $\Sigma^*$ and no other equilibrium.
More precisely:

Bidding $b^2$ in both stages regardless of information is not an equilibrium in the fully-informed model.
The only profitable deviation can be to employ a strategy that uses information, which must be either $\sigma_0$=``play 0 at the first bid, then best respond to information'' or $\sigma_\alpha$=``play $b^2$ at the first bid, then best respond to information''.

\textbf{Is $\sigma_\alpha$ a profitable deviation?}
The best response is to continue bidding $b^2$ if the other bid is $0$ (there are no $v^2$ bidders) and bid $b^3$ if the other bid is $b^2$ (otherwise, bidding $b^2$ again would result in the equilibrium strategy).
To conclude, the best response is getting $\beta$ with probability $1$ rather than splitting $\alpha$ between the two.
Thus, with some implicit conditions on the bids, we get:
\begin{equation}
0<\tfrac{\alpha}{2}<\beta<\alpha<1. \nonumber
\end{equation}
In addition, $\sigma_\alpha$ cannot be an equilibrium in the fully-informed model, and the only possible deviation is bidding $b^3$ in both stages:
\begin{equation}
U_i(\Sigma_\alpha\vert v^2)=(\alpha-\beta)\Pr(X=0)+\beta\E(\frac{1}{X+1})< U_i(b^3,{\Sigma_\alpha}_{-i}\vert v^2)=\frac{\beta}{2}(1+\Pr(X=0)). \nonumber
\end{equation}
We can rearrange these inequalities as inequalities between $\alpha$ and $\beta$ and combine them all:
\begin{equation}\label{Eq: combination of all ineq}
0<\alpha \max\{\frac{1}{2},\frac{\Pr(X=0)}{\E(\frac{1}{X+1})},\frac{\Pr(X=0)}{\frac{1}{2}+\frac{3\Pr(X=0)}{2}-\E(\frac{1}{X+1})}\} < \beta<\alpha\E(\frac{1}{X+1})<1
\end{equation}
In addition, from the properties of expectation, it is clear that $\E(\frac{1}{X+1})<\frac{1}{2}+\frac{\Pr(X=0)}{2}$.
Set $x=\E(\frac{1}{X+1}),y=\Pr(X=0)$.
A necessary condition for \eqref{Eq: combination of all ineq} to hold is
\begin{eqnarray}
\frac{1}{2}<x<\frac{1}{2}+\frac{y}{2} \label{0.2} \\
y<x^2 \label{Eq: 1.1+2.1} \\
y<(\frac{1}{2}+\frac{3}{2}y-x)x \label{Eq: 4.1}
\end{eqnarray}
Summing \eqref{Eq: 1.1+2.1} and \eqref{Eq: 4.1} yields $y<\frac{x}{4-3x}$ which combined with \eqref{0.2} results in $\frac{1}{2}<x<\frac{2}{3}$.
On the other hand, \eqref{Eq: 4.1} is equivalent to  $y(1-\frac{3}{2}x)<\frac{1}{2}x-x^2$.
The right-hand side is negative in the region $\frac{1}{2}<x<\frac{2}{3}$ whereas the left-hand side is positive.
A contradiction -- $\sigma_\alpha$ is not a profitable deviation and $\beta\leq\frac{\alpha}{2}$.\hfill$\bigtriangleup$

\textbf{Can $\sigma_0$ be a profitable deviation?}
The profit with the deviation is:
\begin{equation}
U_i(\sigma_0,\Sigma_{\alpha_{-i}}\vert v^2)=\Pr(X=0)\frac{2}{n}\max\{\alpha,\frac{1}{2}\}+\Pr(X=1)\frac{1}{n-2}\max\{\beta,\frac{\alpha}{2}\} \nonumber
\end{equation}
Our previous case resulted in $\beta\leq\frac{\alpha}{2}$.
In addition, $\frac{1}{n}\leq\alpha$ and $\frac{2\alpha}{n}\leq\alpha$ so $\frac{2}{n}\max\{\alpha,\frac{1}{2}\}\leq\alpha$:
\begin{eqnarray}
U_i(\sigma_0,\Sigma_{\alpha_{-i}}\vert v^2) &\leq& \Pr(X=0)\alpha+\Pr(X=1)\frac{1}{n-2}\frac{\alpha}{2}\leq \nonumber \\ 
&\leq &\Pr(X=0)\alpha+\Pr(X=1)\frac{\alpha}{2}=\alpha(\Pr(X=0)+\frac{1}{2}\Pr(X=1))<  \nonumber \\
&<& \alpha\E(\frac{1}{X+1})=U_i(\Sigma_\alpha\vert v^2) \nonumber
\end{eqnarray}
which means that $\sigma_0$ is not a profitable deviation either.\hfill$\bigtriangleup$

To conclude, it is impossible to satisfy all the conditions in these settings, and if $\Sigma^*$ is not unique without information, it cannot be unique with information.
\hfill \end{proof}

The key component of the proof is that only bidders with private valuation $v^2$ need to be considered and there is a relatively small number of parameters and inequalities that lead to a contradiction.
This proof can be applied to more general cases, with more possible private values and bids, as long as these features remain.

Changing the order to $b^1<b^2<v^1<b^3<v^2$ complicates the problem significantly.
After removing dominated strategies, there are 13 (pure) strategy profiles in the informed model which need to be excluded from being an equilibrium.
This results in a large number of inequalities and ``cases'' to verify, that do not converge to such a neat result as in the Proposition and, more importantly, cannot be easily generalized.

We, therefore, wrote a computer simulation to test our hypothesis in slightly more general settings for different values of $p, n, v^1$ and $M$, the number of possible bids (assuming that they are equally spaced, $B=\{0,\tfrac{1}{M},\ldots,\tfrac{M-1}{M}\}$).
For each $M\in \{4,\ldots,15\}$ and $n\in \{4,\ldots,15\}$, we constructed a grid with $100$ points for $(p,v^1)$ in the domain $[0.1,0.4]\times [0.09, 0.5]$ and calculated all possible pure-strategy equilibria for the chosen parameters.
Despite our efforts, we did not manage to find a counter-example for our conjecture, i.e. a set of parameters for which $\Sigma^*$ is the unique equilibrium with information but not the unique equilibrium without information.

\section{Alternative Assumption Regarding the Joint Distribution of the Private Valuations}\label{app:otherassumptionFn}

A close examination of the proof of Theorem \ref{thmEffUniqueWithLargeN} reveals that the uniform full support distribution can be slightly relaxed.
We chose the uniform full support assumption since it is both general and easy to verify.
Alternative assumptions are more general but harder to verify; moreover, they rely on the bidding strategies and not the distributions.
Here we present one such generalization.

Instead of the uniform full support assumption, assume that there exists $\delta>0$ such that for every $k\in\left\{1,...,K\right\}$, when a bidder has a valuation $v^k$, conditional on $D_k$\footnote{Recall, $D_k$ is when the types of all bidders are at most $v^k$.}
and given that:

\begin{itemize}
\item If each bidder with valuation $v^k$ bids within some range of values with probability $\bar{q}$ at most (independently from the other bidders having valuation $v^k$), then the probability that at most one other bidder will bid within that range is bounded from above by 
\begin{equation}
(1-\bar{q}\delta)^{n-1}+(n-1)\bar{q}\delta(1-\bar{q}\delta)^{n-2}. \nonumber
\end{equation}
This term comes from the best-case scenario for the bidder.
If all the other bidders have the private valuation $v^k$ with the lowest possible probability ($\d$), then the number of $v^k$-type bidders who bid in this range is a random variable $X$ with binomial distribution with parameters $n-1$ and $\d\bar{q}$, and the right-hand side is $\Pr(X\leq 1)$.

\item If we denote by $X$ the number of other bidders with valuation $v^k$, then $\mathbb{E}[\tfrac{1}{X+1}]\leq \tfrac{[1-(1-\delta\bar{q})^n]}{n\delta\bar{q}}$. 
The right-hand side is the expectation when $X$ has the binomial distribution with parameters $n-1$ and $\bar{q}\d$, i.e. when the bidders are independent (see Eq.\eqref{eq expectation}).
\end{itemize}

Note that this assumption is satisfied when valuations are independent with a lower bound of $\delta$ on the probabilities of the values in the support. 
However, some positive correlation typically makes this condition easier to satisfy -- given a bidder's valuation, the probability of bidders with similar valuations should increase.
Negative correlation is also possible.
However in extreme cases, if a bidder's valuation lowers the probability of the other bidders having a similar valuation, then the assumption does not hold.

\section{Continuous Auctions}\label{app:cont_auc}
The main text presented a two-stage auction model with discrete sets of valuations and possible bids.
This mirrors common practice, where currency is discrete and in many cases, auctioneers require ``rounded'' bids (for example, bids may have to be in increments of $10,000\$$).
Moreover, this assumption simplifies the analysis, as the existence of equilibria under all information structures is ensured in this model and the revenue-maximizing strategy profile, $\Sigma^*$, is well defined.
In this Appendix, we provide several results regarding the continuous case.
Although not a full analysis of the problem, our findings point to additional drawbacks of revealing information and strengthen the claim made in the main text.

Our model is identical to the model in Section \ref{Sect:model}, except for the continuity of the bids and valuations.
Hence, we assume that $B=[0,\infty)$ and $V=[v_{min},v_{max}]$ with $0\leq v_{min}<v_{max}$.
We only consider the full information case, and assume that the first-round bids of the two finalists are revealed before the second round.
The tie-breaking rule in the second round is not necessarily symmetric.
For example, it can favor the bidder who made the highest first-round bid.

In this model, the $\beta^j$s are not well defined and therefore neither is the revenue-maximizing strategy profile (or, if defined as the strategy profile in which all bidders bid their valuations in both rounds, it is obviously not an equilibrium).
Instead, we study the existence of an efficient equilibrium, i.e. an equilibrium in which the bidder with the highest valuations wins the item, which in many cases also provides the highest expected payoff to the seller. 
We show that under mild additional assumptions, there are no symmetric efficient equilibria in pure strategies.
This leads us to the conclusion that, as in the discrete model, there are no advantages to a BAFO stage whose sole purpose is to increase auctioneer's revenue, without providing additional information about the item being auctioned.
We note that it is not even clear that an equilibrium exists in this model.
We therefore prove in this Appendix that \emph{if} an equilibrium exists, it is not efficient.

Consider a symmetric efficient equilibrium in pure strategies.
In such an equilibrium there exists a function $f:V\to B$ which represents the first-round bid.
Since the equilibrium is efficient, $f$ must be strictly increasing (otherwise, the highest bidder might not be selected for the second round).
Hence, on the equilibrium path, valuations are revealed by the first-round bids.
The second-round bids can be described by the two functions $g,h:V\times V \to B$, where $g(x,y)$ is the second-round bid of the bidder who bid the highest value $f(x)$ in the first round and $h(x,y)$ is the bid of the runner-up, who bid $f(y)$ in the first round.
The equilibrium is efficient, so $g(x,y)\geq h(x,y)$.

For simplicity, we assume that the valuations are i.i.d. and drawn from $V$ with a probability distribution $p(v)$ that is continuous and positive on $V$.
This ensures that the probability that a valuation falls within a particular interval of (small) length $\epsilon$ is of order $\epsilon$.
The following Lemma provides additional properties of the equilibrium.

\begin{lemma} \label{lem:easyprop}
Let $x$ and $y$ be possible valuations, with $x > y$. Almost surely: 
\begin{enumerate} 
\item $y \leq g(x, y)  \leq x$. \label{lem:easyprop1} 
\item $g(x, y) = \max( f(x), h(x, y))$.\label{lem:easyprop2}
\item If $f(x) < y$, then $h(x, y) = g(x, y) \geq y$. \label{lem:easyprop3}
\item The assumption that $h(x, y) \geq y$ is made without loss of generality. \label{lem:easyprop4}
\item If $x > v_{min}$, then $f(x) < x$. \label{lem:easyprop5}
\end{enumerate}
\end{lemma}
\begin{proof} 
\begin{enumerate}
\item If $g(x,y) > x$, the player with valuation $x$ is bidding above his value and should slightly decrease his second-stage bid; if $g(x,y) < y$, the player with valuation $y$ should deviate and bid slightly above $g(x,y)$. 
\item The auction rules prevent bidding below the first-round bid, so $g(x,y) \geq f(x)$.
Clearly, $g(x,y) \geq h(x,y)$ because the auction is assumed to be efficient, so $g(x,y) \geq  \max( f(x), h(x, y))$. 
The reverse inequality holds because otherwise the winner should decrease his second-round bid. 
\item Immediate from \ref{lem:easyprop1} and \ref{lem:easyprop2}: $g(x,y)\geq y >f(x)$ so $g(x,y)\neq f(x)$ and the only remaining possibility is $g(x,y)=h(x,y)$.  
\item If we replace $h$ by $\tilde{h}$ defined by $\tilde{h}(x,y) = \max(y, h(x,y))$, then we still have an equilibrium. 
Indeed, this does not change anything when $f(x) < y$ by point \ref{lem:easyprop3}), and when $f(x) > y$, this does not change the equilibrium payoffs nor the possible deviations and their payoffs.
Therefore, we can assume that the equilibrium we have in hand is the one where $h(x,y)\geq y$. 

\item If $f(x)=x$ and $x > v_{min}$, then the bidder obtains a zero payoff though he could obtain a positive payoff by playing as if he had a slightly lower valuation.
\end{enumerate}
\hfill\end{proof}

For our proof, we still need an extra assumption.
Several natural assumptions are possible.
We choose to assume that the runner-up never bids above his value: 

\begin{assumption} \label{ass}
For all $x,y\in V$ with $x > y$, we assume $h(x,y) \leq y$.
\end{assumption}

Combined with Lemma \ref{lem:easyprop} (part \ref{lem:easyprop4}), on equilibrium path and without loss of generality, $h(x,y)=y$ and by part \ref{lem:easyprop2}, $g(x,y)=\max (f(x),y)$.
It follows that $g$ is non-decreasing in each of its arguments and strictly increasing if both arguments increase: $x<x'$ and $y<y'$ imply $g(x,y)<g(x',y')$.
We are now ready to show that in this model, an efficient equilibrium does not exist, i.e. that  the above properties cannot all hold together.

\begin{proposition}
In a two-stage BAFO auction with a continuous set of bids and valuations, there are no symmetric efficient equilibria in pure strategies when all bids are revealed to the bidders before the BAFO stage. 
\end{proposition}
\begin{proof} Assume by contradiction that such an equilibrium exists and consider bidder $n$ with valuation $z>v_{min}$.
Let $x \geq y$ be the two highest  valuations among the other $n-1$ bidders.
Almost surely, $x > y$, which we now assume.
Let $\epsilon>0$ and consider a deviation by bidder $n$ that consists of bidding $f(z-\epsilon)$ in the first round instead of $f(z)$ and then playing optimally in the second round.
We study the gain from this deviation in all possible cases (with non-zero probability) and show that it is positive for $\epsilon$ small enough. 
The main two cases are the case where bidder $n$ does not have the highest valuation (``case A'': $z<x$) and where bidder $n$ has the highest valuation (``case B'': $z>x$).
The full analysis of all sub-cases that occur with positive probability is provided here and summarized in Table \ref{tab:sum}.

\textbf{Case A: $z<x$}

\begin{itemize}
\item [Case A.1:] $f(x)<z-\epsilon<z<x$.

 As $\epsilon$ goes to $0$, this case occurs with probability of order $1$ (bounded away from zero).
Bidder $n$'s equilibrium payoff is $0$, as the highest bidder wins the auction.
If bidder $n$ deviates, he is still selected and his opponent's second-round bid is $g(x,z-\epsilon)=h(x,z-\epsilon)$, so the deviation payoff is arbitrarily close to $\max (0, z-h(x,z-\epsilon))$, which is equal to $\epsilon$ under Assumption \ref{ass}.
Thus, the gain is of order $\epsilon$.
\item [Case A.2] Otherwise.

The payoff from the equilibrium strategy is $0$ and the payoff when deviating is non-negative, so the deviation gain is non-negative.
\end{itemize}

\textbf{Case B: $z>x$}
\begin{itemize}
\item [Case B.1] $z-\epsilon>x$. 

This occurs with probability of order 1.
The equilibrium payoff is $z-g(z,x)$ and the payoff when deviating is $z-g(z-\epsilon,x)$.
The gain from deviation is therefore 
\[g(z,x)-g(z-\epsilon,x)=\max (f(z),h(z,x))-\max (f(z-\epsilon),h(z-\epsilon,x))\]
under Assumption \ref{ass} this results in $\max (f(z),x)- \max (f(z-\epsilon),x)) \geq 0$.
\item [Case B.2] $y<z-\epsilon<x<z$.

This occurs with probability of order $\epsilon$.
The equilibrium payoff is $z-g(z,x)\leq z-x$, whereas when deviating, the payoff is arbitrarily close to $z-g(x,z-\epsilon)\geq z-x$.
Thus, there is at most an arbitrarily small loss from deviating.
Under Assumption \ref{ass}, the gain is actually strictly positive: 
\[g(z,x)-g(x,z-\epsilon)=\max (f(z),x)- \max (f(x),z-\epsilon)\geq x-\max (f(x),z-\epsilon)> 0. \]
\item [Case B.3] $z-\epsilon<y<x<z$.

This occurs with probability of order $\epsilon^2$.
The payoff when deviating is zero while the payoff in equilibrium is $z-g(z,x)\leq z-x \leq \epsilon$.
Hence, by deviating, the bidder loses at most $\epsilon$.
\end{itemize}
To summarize, under Assumption \ref{ass}, by bidding $f(z-\epsilon)$ instead of $f(z)$, the bidder can obtain a positive profit of order $\epsilon$ and a loss of order $\epsilon^3$.
Hence, for $\epsilon$ small enough, this is a profitable deviation and the functions $f,g,h$ do not form a symmetric efficient equilibrium.
\hfill\end{proof}
\begin{table}
\begin{center}
\begin{tabular}{|l|c|c|}
\hline
Case 						& Order & Deviation gain \\
\hline
 A1. $f(x)<z-\epsilon<z<x$ & 1 		& $\varepsilon$ \\
 \hline
 A2. Other cases with $z < x$ 			& 1 			& $\geq 0$ 					\\
 \hline
 B1. $z-\epsilon > x$ 				& 1 			&  $g(z, x) - g(z-\epsilon, x) \geq 0$  \\
 \hline 
 B2. $y < z-\epsilon < x < z$ 		& $\epsilon$ & $g(z, x) - g(x, z-\epsilon) > 0$ 	 \\
 \hline
 B3. $z-\epsilon < y < x < z$ 		& $\epsilon^2$ & $g(z,x) - z \geq x - z \geq - \epsilon$ \\ 
 \hline
\end{tabular}
\label{tab:sum}
\end{center}
\caption{Summary of the 5 cases under Assumption \ref{ass}.} 
\end{table}

This proposition shows that with information, equilibria, if they exist, cannot be efficient.
Contrastingly, without information the auction is equivalent to a first-price (single-stage) sealed-bid auction (the proof of Lemma \ref{thmNoInfoIsFPSA} is also valid in this model), so there exists an efficient equilibrium. 
Although efficiency does not directly imply best revenue for the seller, it is often the case and the lack of efficiency suggests that there might be disadvantages to revealing information.
We conclude that there are good reasons not to reveal information between the stages in the continuous case as well, which reinforces the paper's findings under the discrete model.

\end{document}